\documentclass[11pt]{article}
\usepackage{amsmath,amssymb,amsfonts,amsthm,epsfig,verbatim}
\usepackage{times}
\usepackage{color}
\usepackage{framed}

\newif\ifhyper\IfFileExists{hyperref.sty}{\hypertrue}{\hyperfalse}
\hypertrue
\ifhyper\usepackage{hyperref}\fi

\newenvironment{proofof}[1]{\par{\noindent \bf Proof of #1:}}{\qed\par}
\newcommand{\red}[1]{{ {#1}}}
\newtheorem{theorem}{Theorem}

\newtheorem{lemma}[theorem]{Lemma}
\newtheorem{proposition}[theorem]{Proposition}
\newtheorem{corollary}[theorem]{Corollary}
\newtheorem{claim}[theorem]{Claim}
\newtheorem{fact}[theorem]{Fact}

\newtheorem{definition}[theorem]{Definition}

\newenvironment{sketch}{\noindent \emph{Proof
sketch:}\nopagebreak[2]}{\qed \medskip}

\newcommand{\Var}{\operatorname{Var}}

\newcommand{\feas}{\mathrm{Feas}}

\newcommand{\ignore}[1]{}

\newcommand{\cref}[1]{Corollary~\ref{cor:#1}}

\newcommand{\dtv}{\mathrm{d_{TV}}}
\newcommand{\dcdf}{\mathrm{d_{cdf}}}

\newcommand{\bX}{\mathbf{X}}
\newcommand{\bY}{\mathbf{Y}}
\newcommand{\bZ}{\mathbf{Z}}

\renewcommand{\Pr}{\operatorname{{\bf Pr}}}

\date{}

\author{
Anindya De\footnote{Northwestern University. Email: \texttt{ anindya@eecs.northwestern.edu.}}}

\begin{document}

\setcounter{page}{0}

\title{Boolean function analysis meets stochastic optimization: An approximation scheme for stochastic knapsack}



\maketitle

\thispagestyle{empty}

\begin{abstract}
The stochastic knapsack problem is the stochastic variant of the classical knapsack problem in which the algorithm designer is given a a knapsack with a given capacity and a collection of items where each item is associated with a  profit and a probability distribution on its size. The goal is to select a subset of items with maximum profit and violate the capacity constraint with probability at most $p$ (referred to as the overflow probability). 

While several approximation algorithms~\cite{KRT97, GI99, BGK11, DeanGV08, LiY13} have been developed for this problem, most of these algorithms relax the capacity constraint of the knapsack. In this paper, we design efficient approximation schemes for this problem without relaxing the capacity constraint.  
 
\begin{itemize}
\item[(i)]Our first result is in the case when item sizes are Bernoulli random variables. In this case, we design a (nearly) fully polynomial time approximation scheme (FPTAS) which only relaxes the overflow probability.
\item[(ii)] Our second result generalizes the first result to the case when all the item sizes are  supported on a (common) set of constant size. In this case, we obtain a quasi-FPTAS. 
\item[(iii)] Our third result is in the case when item sizes are so-called ``hypercontractive" random variables i.e., random variables whose second and fourth moments are within constant factors of each other. In other words, the kurtosis of the random variable
is upper bounded by a constant. 
 This class has been widely studied in probability theory and most natural random variables are hypercontractive including well-known families  such as Poisson, Gaussian, exponential and Laplace distributions. In this case, we design a polynomial time approximation scheme which relaxes both the overflow probability and maximum profit. 
\end{itemize}

Crucially, all of our algorithms meet the capacity constraint exactly, a result which was previously known only when the item sizes were Poisson or Gaussian random variables~\cite{GI99, Goyal:2010}. \red{Our results
rely on new connections between Boolean function analysis and  stochastic optimization and  are obtained by an adaption and extension of ideas  such as (central) limit theorems, moment matching theorems and the influential critical index machinery of Servedio~\cite{Servedio:07cc} developed in the context of complexity theoretic analysis of halfspaces.} We believe that \red{these ideas and techniques} may prove to be useful in other stochastic optimization problems as well.  

\end{abstract}

\newpage

\section{Introduction}
The knapsack problem is one of the most well-studied combinatorial optimization problems~\cite{Garey:1979} 
and early work on this problem dates back more than a century~\cite{GBM:96}. While several variants of this problem have now been studied, in its simplest instantiation, we are given a set of items, each associated with a size and profit. Given a capacity constraint of $C$, the task is to find a subset of items which maximizes the total profit and whose total size is bounded by $C$. While the knapsack problem is known to be NP-hard, it admits both a pseudopolynomial time algorithm as well as a fully polynomial time approximation scheme, thus making the problem \emph{tractable} in many settings. 

In this paper, we are interested in the stochastic variant of this problem. Here, the item sizes are no longer fixed and are instead given as a probability distribution (supported on $\mathbb{R}^+$, i.e.,~the set of positive real numbers). As is the case with nearly any combinatorial optimization problem, there are several potential stochastic variants of the knapsack problem which have been studied in the literature. See~\cite{DeanGV08, BGK11, LiY13, gupta2011approximation} for a partial list of results in various types of models. Our emphasis is on the so-called \emph{chance-constrained} version of knapsack (alternately referred to as the fixed-set version of stochastic knapsack). A problem instance here is given by items $I_1, \ldots, I_n$ where each $I_j = (\bX_j, v_j)$. Here $\{\bX_j\}$ are (independent) non-negative real-valued random variables representing the \emph{stochastic size} of each item and $v_j$ are non-negative real numbers representing the profit of each item.  Given a knapsack capacity $C$ and an overflow probability $p \ge 0$, the aim is to find a set $S \subseteq [n]$ of items which  maximizes $\sum_{j \in S} v_j$ such that  $\Pr[\sum_{j \in S} \bX_j >C] \le p$. The second condition, namely $\Pr[\sum_{j \in S} \bX_j >C] \le p$ is equivalent to saying that the constraint on the knapsack is violated with probability at most $p$.

More generally, in a chance-constrained optimization problem, we want to maximize an objective function while allowing the constraints to be violated with a maximum probability $p$ (which is referred to as the unreliability level). Such problems have been long investigated  in the optimization community starting with the work of Charnes \emph{et al.}~\cite{CCS:58} and the seminal work of Pr\'{e}kopa~\cite{Prekopa:70, Prekopa:95} and continue to remain the topic of current research~\cite{BNR02, NS06}. See the books
\cite{BEN:09, Shapiro:2014} which provide a good survey of the current state of the art of this family of problems. \red{As far as the author is aware, work in the TCS community has mostly focused on specific problems in this family (as opposed to developing a broad theory of chance-constrained optimization problems). However, some authors (see Nikolova's thesis~\cite{Nikolova:09}) have considered  ways of modeling risk other than  via chance constraints.}

Before we discuss prior work on this problem, let us focus on some core issues of the stochastic knapsack problem which this paper seeks to address. \begin{enumerate}
\item Given $S \subseteq [n]$, it is $\mathsf{\#P}$ hard to compute $\Pr[\sum_{i \in S} \bX_j>C]$ even for very simple classes of random variables (such as when each $\bX_j$ is $\pm w_j$ with probability $1/2$ each). Thus, even checking whether a given solution meets the probabilistic constraint exactly is computationally hard. 
\item For the usual (i.e., deterministic) knapsack problem, when $\{v_j\}$ are arbitrary non-negative numbers, maximizing the profit is NP-hard. 
\end{enumerate}
This suggests  that relaxing (at least one of) the overflow probability or the maximum profit is necessary to obtain efficient algorithms.\footnote{We do not know of any formal hardness results here apart from those trivially implied by the hardness of the deterministic knapsack problem.} However, with the exception of two cases,  namely when $\{\bX_i\}$ are distributed as Poisson~\cite{GI99} or Gaussian~\cite{Goyal:2010}, all known algorithms relax the capacity constraint as well. In fact, both these algorithms rely on very delicate properties of these distributions: (i) sum of two Poisson (resp. Gaussian) random variables is a Poisson (resp. Gaussian) random variable. (ii) their distribution is determined entirely by (at most) their first two moments. In fact, these algorithms cannot handle the case when some of the variables follow a Gaussian and the others follow a Poisson distribution.

The main focus of this work is to obtain approximation schemes for stochastic knapsack without relaxing the capacity constraint for a large class of random variables. In particular, we obtain such approximation schemes in three different settings for stochastic knapsack: 
\begin{itemize}
\item[(i)] When $\{\bX_j\}$ are Bernoulli random variables, we obtain a $\mathsf{poly}(n) \cdot \mathsf{quasipoly}(1/\epsilon)$ time approximation scheme. 
\item[(ii)] When $\{\bX_j\}$ are all supported on a common support $\{a_1, \ldots, a_k\}$, we obtain a 
$(nk/\epsilon)^{O(k \log (1/\epsilon))^{k+1}}$ time approximation scheme. Note that when $k=O(1)$, the running time is quasipolynomial in $n$ and $\epsilon$. 
For $k=2$, this is the same result as the first one (with a slightly worse running time). 
\item[(iii)] When $\{\bX_j\}$ are so-called $(c,2,4)$  hypercontractive random variables, we obtain a $n^{\tilde{O}(c^4/\epsilon^2)}$ time approximation scheme.  Roughly speaking, a random variable is $(O(1), 2,4)$ hypercontractive if its (central) fourth moment is bounded by the square of its variance up to a $O(1)$ factor. In the language of statistics, this is also referred to as having the kurtosis bounded by a constant. 
While we later  elaborate on this  notion later, we  mention here that most \emph{common random variables} are $(c,2,4)$ hypercontractive for a constant $c$. Examples include Poisson\footnote{We note that a Poisson with mean $\lambda$ is $(\lambda^{-1/4}, 2,4)$ hypercontractive. Thus, the constant $c$ is upper bounded by $O(1)$ only when the mean of the Poisson is bounded away from $0$. On the other hand, Gaussian, exponential, Laplace etc. are $(c,2,4)$ hypercontractive for a fixed $c>0$. }, Gaussian, exponential, Laplace, uniform on an interval, finitely supported distributions etc. Just to contrast with our earlier remark, our algorithms can easily handle the case when, say, some of the $\{\bX_i\}$ follow a Poisson distribution while others follow a  Gaussian distribution. 
\end{itemize}

\subsection{Our results}
We now formally state our results. 
To do this, we begin by formally defining an instance of the stochastic knapsack problem. 
\begin{definition}~\label{def:knapsack}
An instance of the stochastic knapsack problem is specified  by a list of items $\{(\bX_i, v_i)\}_{i=1}^n$, a capacity $C>0$ and a risk budget $p>0$. Here each $v_i$ is a positive rational number representing the profit of item $i$ and $\bX_i$ is a non-negative random variable representing the ``stochastic size" of the item. 
For $q>0$, let $\feas_q \subseteq 2^{[n]}$ be defined as
$$
\feas_q = \{S \subseteq [n]: \Pr[\sum_{j \in S} \bX_j >C ] \le q\}. 
$$
The task here is to output $S \subseteq \feas_p$ such that 
$$
\sum_{i \in S} v_i = \max_{\tilde{S} \in \feas_p} \sum_{i \in \tilde{S}} v_i. 
$$
Let $\mathcal{D}$ be a class of non-negative real valued random variables and $\mathcal{V}\subseteq \mathbb{R}^+$. If $\{\bX_i\}_{i=1}^n \subseteq \mathcal{D}$ and $\{v_i\}_{i=1}^n \subseteq \mathcal{V}$, then we say that it is an instance of \emph{type $(\mathcal{D}, \mathcal{V})$.}
\end{definition} 

\begin{definition}
Given an instance of stochastic knapsack as in Definition~\ref{def:knapsack}, we say that an algorithm outputs an $(\epsilon,0)$ approximation if it outputs $S \in \feas_{p+\epsilon}$ such that 
$
\sum_{i \in S} v_i \ge \max_{\tilde{S} \in \feas_{p}} \sum_{i \in \tilde{S}} v_i. 
$ An algorithm is said to output an $(\epsilon, \epsilon)$ approximation if it outputs $S \in \feas_{p+\epsilon}$ such that 
$
\sum_{i \in S} v_i \ge (1-\epsilon) \cdot \max_{\tilde{S} \in \feas_{p}} \sum_{i \in \tilde{S}} v_i. 
$
\end{definition}
Thus, in an $(\epsilon,0)$ approximation algorithm, we only relax the overflow probability (by an additive $\epsilon$) whereas an $(\epsilon, \epsilon)$ approximation algorithm relaxes both the profit and the 
overflow probability. Note that crucially neither type of approximation relaxes the capacity constraint. 

\subsubsection{Approximation scheme for Bernoulli random variables}
Our first result is an $(\epsilon,0)$ approximation algorithm when $\{\bX_i\}$ are Bernoulli random variables. More formally, let $\mathcal{D}_B$ be the class of Bernoulli random variables and let $\mathbb{Q}^+$ be the set of positive rational numbers. Then, we have the following theorem. 
\begin{theorem}~\label{thm:Bernoulli1}
There is an $(\epsilon,0)$ approximation algorithm for stochastic knapsack instances of type $(\mathcal{D}_B, \mathbb{Q}^+)$ running in time $\mathsf{poly}(n) \cdot (1/\epsilon)^{\log^2(1/\epsilon)}$. 
\end{theorem}
\red{Our theorem make significant use of results on (central) limit theorems for sums of Bernoulli random variables\footnote{They are similar in flavor but technically very different from say the well-known Berry-Ess\'{e}en theorem} (aka Poisson binomial distributions)} which have recently been a subject of investigation in computational learning theory~\cite{DDS12stoc, DDOST13} and algorithmic game theory~\cite{DP09, DP08}. In particular, such limit theorems (approximately) characterize the distribution of sums in terms of their low-order moments. Combining this with standard dynamic programming techniques gives us the algorithm. It is useful  to mention here that while the specific probabilistic techniques we use here are new (for this line of work), dynamic programming as an algorithmic tool has been a staple in several papers in this area~\cite{GI99, BGK11, LiY13}. 

\subsubsection{Approximation scheme for $k$-supported random variables}

Our second result is an $(\epsilon,0)$ approximation algorithm when $\{\bX_i\}$ are independent random variables, all supported on a common set $A = \{a_1, \ldots, a_k\}$. More formally, given any set $A$ of size $k$, let $\mathcal{D}_{A}$ be the set of random variables supported on $A$. Then, our result is the following. 
\begin{theorem}~\label{thm:Bernoulli2}
There is an $(\epsilon,0)$ approximation algorithm for stochastic knapsack instances of type $(\mathcal{D}_A, \mathbb{Q}^+)$ running in time $(nk/\epsilon)^{O(k \log k + \log(1/\epsilon))^{k+1}}$ for any set $A$ of size $k$.  
\end{theorem}
As with Theorem~\ref{thm:Bernoulli1}, this algorithm also makes use of very recent moment-matching theorems for so-called Poisson multinomial distributions~\cite{daskalakis2015structure} coupled with standard dynamic programming techniques. Note that  after a suitable translation, any set $|A|$ of size $2$ can be assumed to be $\{0,1\}$. Thus, Theorem~\ref{thm:Bernoulli1} is a special case of Theorem~\ref{thm:Bernoulli2} with a faster running time.

\subsubsection{Approximation scheme for hypercontractive random variables}
Our next result is for a much broader, albeit incomparable class of random variables namely hypercontractive random variables. This is a very widely studied class of random variables in Boolean function analysis (see O'Donnell's book~\cite{ODonnell:2014}). We begin with some brief motivation and definitions. Let us begin with the definition of central moments. 
\begin{definition}~\label{def:moments}
For a real-valued random variable $\bX$ and for $j>1$, we define $\mu_j(\bX) = \mathbf{E}[|\bY|^j]$ where $\bY$ is the random variable $\bY = \bX - \mathbf{E}[\bX]$. In other words, for $j>1$, $\mu_j(\bX)$ is the $j^{th}$ absolute central moment of $\bX$. 
\end{definition}
Note that  $\mu_2(\bX)$ is simply the variance of $\bX$. Now, by Jensen's inequality, it easily follows that for any $j \ge 2$, we have $\mu_j(\bX) \ge (\mu_2(\bX))^{j/2}$. 
Essentially, a real-valued random variable is  said to be hypercontractive if the inequality holds in the opposite direction (with appropriate constants). More formally, we define the notion of $(c,2,4)$ hypercontractivity below. 
\begin{definition}~\label{def:hyper}
A real-valued random variable  $\bX$ is said to be $(c,2,4)$-hypercontractive if 
$\mu_4(\bX) \le c^4 \cdot \mu_2^2(\bX)$.   
\end{definition}
In the language of statistics, this is equivalent to stating that the kurtosis of $\bX$ is at most $c^4$.
The notion of $(c,2,4)$ hypercontractivity is the quantitative analogue of the existence of fourth moment of $\bX$ (provided the second moment exists). 
We refer to $c$ as the hypercontractivity constant for $\bX$. 
As we have said before, most well-known families of random variables such as Poisson, Gaussian, Laplace and exponential random variables are $(O(1),2,4)$ hypercontractive. For the convenience of the reader, in Appendix~\ref{app:hyper}, we list some common families of random variables which are  $(c,2,4)$ hypercontractive along with the (corresponding) explicit value of $c$.

On the other hand, there are real-valued random variables which are not $(c,2,4)$ hypercontractive for any $c$. For example, consider the random variable $\bX$ supported on $[-1,1]$ where the density of $\bX$ is given by $\bX(x) = |x|^{-\frac13}$. While $\mu_2(\bX) =6$, $\mu_4(\bX)$ is unbounded and is thus, not $(c,2,4)$ hypercontractive for any $c$. \red{We mention  that the definition of hypercontractivity we use here is a  weakening of the more standard notion of hypercontractivity from  analysis~\cite{Krakowiak1988, wolff2007hypercontractivity, MOO10}.  The latter definition  is nicer from an analysts' point of view but we prefer the definition here for two reasons: (a) 
it is  more intuitive to understand. (b) For our application, this definition is easier to work with and in fact, given a random variable $\bX$, it is easier to check our condition of hypercontractivity. 
}

For now, we state our main result for stochastic knapsack when $\{\bX_i\}$ are $(c,2,4)$ hypercontractive random variables. Let $\mathcal{D}_c$ be the class of non-negative $(c,2,4)$ hypercontractive random variables. Our main result is the following. 
\begin{theorem}~\label{thm:hyper-c}
There is an $(\epsilon, \epsilon)$ approximation algorithm for stochastic knapsack instances of type $(\mathcal{D}_C, \mathbb{Q}^+)$ running in time $n^{\tilde{O}(c^4/\epsilon^2)}$. 
\end{theorem}
Thus for $(c,2,4)$ hypercontractive random variables, our approximation algorithm relaxes both the profit as well as the overflow probability. Crucially, our algorithm does not relax the capacity constraint of the knapsack. We now highlight an important corollary of this theorem. Namely, let us say a finitely supported distribution $\bX$ is $\alpha$-balanced if $\min_{x: \bX(x) \not = 0} \bX(x) = \alpha$ i.e., the probability of any support point is at least $\alpha$. In Proposition~\ref{prop:hyper-finite}, we prove that any $\alpha$-balanced distribution is $(\alpha^{-1/4}, 2,4)$-hypercontractive. As a corollary of Theorem~\ref{thm:hyper-c}, we get an approximation scheme for stochastic knapsack when $\{\bX_i\}$ are $\alpha$-balanced which runs in time $n^{O((\alpha \epsilon^2)^{-1})}$. Note that every finitely supported distribution is $\alpha$-balanced for some fixed $\alpha>0$. Thus, this implies the following corollary. 
\begin{corollary}~\label{corr:finite}
There is an $(\epsilon,\epsilon)$ approximation algorithm for stochastic knapsack instances where the random variables $\{\bX_i\}$ are finitely supported. Further, the running time is  $n^{O((\alpha \epsilon^2)^{-1})}$ where $\alpha$ is defined as $\alpha = \mathop{\min}_{i, x: \bX_i(x) \not =0} \bX_i(x)$. 
\end{corollary}
This should be compared to the results of \cite{GI99, ChekuriK05} (which we discuss shortly) where they obtained  a polynomial time approximation scheme for the case when each $\{\bX_i\}$ is supported on $0$ and another point (this point can depend on $i$). They call such random variables ``Bernoulli-type" random variables. 
On one hand, these papers obtain a fully polynomial time approximation scheme in this setting  whereas our running time depends on the ``balanced-ness" parameter of the random variables. 
On the other, the algorithm in \cite{GI99, ChekuriK05}   relaxes the capacity constraint, whereas ours does not and in fact, ours yields an efficient approximation scheme for any constant sized support. \red{Further, our algorithms apply to a much broader class of random variables, and are not tailored towards Bernoulli-type random variables.} 

Theorem~\ref{thm:hyper-c} follows by a reduction to the following theorem which obtains an $(\epsilon,0)$ approximation when the profit of each item is a polynomially bounded integer. Let $\mathbb{Z}_M^+$ be the set of non-negative integers bounded by $M$. We prove the following theorem.  
\begin{theorem}~\label{thm:hyper-main-bounded} 
There is an $(\epsilon,0)$ approximation algorithm for stochastic knapsack instances of type $(\mathcal{D}_C, \mathbb{Z}_M^+)$ running in time $n^{\tilde{O}(c^4/\epsilon^2)} \cdot \mathsf{poly}(M)$.  
\end{theorem}
The reduction from Theorem~\ref{thm:hyper-c} to Theorem~\ref{thm:hyper-main-bounded} is essentially the standard reduction that yields a polynomial time approximation scheme for (the standard) knapsack by reducing to knapsack with polynomially bounded weights.
Thus, our focus will essentially be on proving Theorem~\ref{thm:hyper-main-bounded}. We now give a brief description of prior work followed by a high level overview of our techniques. 

\subsection*{Prior work} 
Motivated by the problem of  allocating bandwidth to bursty connections, Kleinberg, Rabani and Tardos~\cite{KRT97} were the first to study the stochastic knapsack problem. They proved several incomparable results for the case when $\{\bX_i\}$ are Bernoulli type random variables\footnote{Recall that $\bX_i$ is said to be Bernoulli-type if its support  size is at most $2$ and one of the points in the support is $0$.}. In particular they obtained a $\log(1/p)$ approximation without relaxing either the capacity constraint or the overflow probability (where $p$ is the overflow probability). They also obtained a $O(1/\epsilon)$ approximation by either relaxing the overflow probability to $p^{1-\epsilon}$ or the capacity constraint by a factor of $(1+\epsilon)$. 
Soon thereafter, Goel and Indyk~\cite{GI99} studied this problem for Poisson, exponential and Bernoulli-type random variables and obtained a PTAS for the first two and a quasi-PTAS for the last one (this was improved subsequently to a PTAS 
  by Chekuri and Khanna~\cite{ChekuriK05}). The main caveat of their result was that in addition to relaxing the overflow probability, for both exponential and Bernoulli-type distributions, they relaxed the capacity constraint \red{by a multiplicative factor of $(1+\epsilon)$} as well. 
  
  Subsequently, there were several papers which explored other models of stochastic knapsack, particularly, the power of adaptive strategies in this context~\cite{DeanGV08, BGK11, gupta2011approximation, dean2005adaptivity, LiY13}. In terms of progress on the fixed set version (considered in this paper), Goyal and Ravi~\cite{Goyal:2010} obtained a PTAS when the item sizes are Gaussian.   Finally, Bhalgat, Goel and Khanna~\cite{BGK11} obtained a PTAS which works for any random variable but relaxes all the three parameters, namely  the capacity constraint, optimal value and overflow probability by a factor of $(1+\epsilon)$. The running time in \cite{BGK11} was $n^{O_\epsilon(1)}$ where $O_{\epsilon}(1)$ is doubly exponential in $\epsilon$. This was improved to a singly exponential in $\epsilon$ by Li and Yuan~\cite{LiY13} (using different techniques). To summarize, the results of \cite{BGK11, LiY13} essentially settle the case of stochastic knapsack if one is willing to relax the capacity constraint. However, without relaxing the capacity constraint, we knew  of approximation schemes in precisely two cases: When the item sizes are Gaussian~\cite{Goyal:2010} or when they are Poisson~\cite{GI99}.
  \red{In fact, prior to this work the best known algorithm that does not violate capacity constraints achieved a $O(\log(1/p))$ approximation to the objective even when the item sizes are Bernoulli random variables.}
  
  \subsection*{Overview of our techniques}  
  
  \textbf{Proof overview of Theorem~\ref{thm:Bernoulli1} and Theorem~\ref{thm:Bernoulli2}:}  At a very high level, there are two main technical ideas in this paper. We start with the first  idea, which is used to prove Theorem~\ref{thm:Bernoulli1} and Theorem~\ref{thm:Bernoulli2} and is significantly easier to explain. 
  The main plan is to exploit limit theorems for sums of independent random variables (of the type appearing in Theorem~\ref{thm:Bernoulli1} and Theorem~\ref{thm:Bernoulli2}). 
  In a nutshell, these limit theorems approximately characterize the distribution of the sum in terms of its low order moments. This characterization is then used to convert the stochastic knapsack problem into a deterministic multidimensional knapsack problem. However, we know of pseudopolynomial time algorithms for the latter which translates into an approximation scheme for  stochastic knapsack.  In fact, the idea of converting stochastic knapsack into multidimensional deterministic knapsack (via different means) can be traced back to the work of Goel and Indyk~\cite{GI99}.
  \red{The novel aspect of our work here is the use of ideas and tools from limit theorems to perform this conversion.} 
  
  To explain the idea in a little more detail, let us first focus on Theorem~\ref{thm:Bernoulli1} (i.e.,~Bernoulli sized items). Consider a stochastic knapsack instance (as in Definition~\ref{def:knapsack}) where the items are $\{(\bX_\ell, v_\ell)\}_{\ell=1}^n$, capacity is $C$ and the overflow probability is $p$. Also, assume that all the probabilities in $\{\bX_{\ell}\}$ are rounded to the nearest multiple of $\epsilon/ (4n)$. This can be accomplished by losing at most an additive $\epsilon/4$ in the overflow probability of any subset of $[n]$.
  Now, assume that  $S_{opt} \subseteq [n]$ is the optimal solution to this problem. Consider the random variable $\bZ_{S_{opt}} = \sum_{\ell \in S_{opt}} \bX_\ell$. Our algorithm is split into two cases: (i) when  $\mathsf{Var} (\bZ_{S_{opt}}) \ge 1/\epsilon^2$ and (ii) when  $\mathsf{Var} (\bZ_{S_{opt}}) < 1/\epsilon^2$.  
In the first case, our algorithm finds a set $S \subseteq [n]$ such that for $\bZ_{S} = \sum_{\ell \in S} \bX_\ell$, its first two moments  are the same as 
$\bZ_{S_{opt}}$ and $\sum_{\ell \in S} v_\ell \ge \sum_{\ell \in S_{opt}} v_\ell$ (this can be accomplished by dynamic programming). 
Note that while we do not know the values of the first two moments of $\bZ_{S_{opt}}$, there are only $\mathsf{poly}(n/\epsilon)$ possibilities for these as all the probabilities in $\{\bX_\ell\}$ are integral multiples of $\epsilon/(4n)$. Thus, we can exhaustively try out all possibilities and find out a set $S$ for each possibility. The key fact that we use here is a so-called discrete central limit theorem for sums of Bernoulli random variables (Lemma~\ref{lem:CLT1}): 
Namely, if the first two moments of $\bZ_{S_{opt}}$ and $\bZ_S$ are the same (and the variance is at least $1/\epsilon^2$), then they are $\epsilon$-close to each other in total variation distance. Thus, the overflow probability of $\bZ_S$ is at most $\epsilon$ more than $\bZ_{S_{opt}}$. This finishes the first case. 

The algorithm in the second case is quite similar to the first case but here we find a set $S$ such that the first $O(\log(1/\epsilon))$ moments of $\bZ_S$ match those of $\bZ_{S_{opt}}$ (instead of just the first two moments as was done in case (1)). The key fact on which we rely here is a recent so-called ``moment matching theorem" of Daskalakis and Papadimtriou~\cite{DP09} (Lemma~\ref{lem:moment-matching}) which essentially says that matching $O(\log(1/\epsilon))$ moments implies $\epsilon$-closeness in total variation distance between $\bZ_S$ and $\bZ_{S_{opt}}$ (we are glossing over an additional technical condition required to apply this theorem and indeed our algorithm is also somewhat more involved). In fact, naively applying Lemma~\ref{lem:moment-matching} results in a running time of $(n/\epsilon)^{\log^2(1/\epsilon)}$ in Case (2). To instead get a running time of $\mathsf{poly}(n) \cdot (1/\epsilon)^{\log^2(1/\epsilon)}$ (as claimed in Theorem~\ref{thm:Bernoulli1}), some additional complication is required and one instead has to apply a so-called Poisson approximation theorem~\cite{LeCam:60} in tandem with the moment matching theorem of \cite{DP09}. 

We do not discuss the proof for Theorem~\ref{thm:Bernoulli2} here but at a high level, it also relies on a ``moment-matching theorem" similar to case (ii) of Theorem~\ref{thm:Bernoulli1}. In particular, we use a very recent ``moment-matching theorem" for so-called Poisson multinomial distributions (PMDs) due to Daskalakis, Kamath and Tzamos~\cite{daskalakis2015structure}. The actual theorem statement is somewhat more complicated, so we refrain from discussing it here any further.  
We also mention that Li and Yuan~\cite{LiY13} had used a similarly flavored idea for stochastic knapsack: Namely, they used a so called compound Poisson approximation~\cite{LeCam:60} to convert stochastic knapsack into (deterministic) multidimensional   knapsack. While their method of approximation is quite general and in fact applies to any random variable, their guarantee is  weaker  and in fact, they relax the capacity constraint even when all the sizes $\{\bX_i\}$ are Bernoulli random variables.

\textbf{Proof overview of Theorem~\ref{thm:hyper-c}:}  As we have said, the proof of Theorem~\ref{thm:hyper-c} essentially reduces to proving Theorem~\ref{thm:hyper-main-bounded} i.e., where the profits are non-negative integers bounded by $M$. So, let us focus on the proof of  Theorem~\ref{thm:hyper-main-bounded}. 
\red{The main idea here is a new connection between the stochastic knapsack problem and the structural analysis of halfspaces.}
We begin by recalling that halfspaces are Boolean functions $f: \mathbb{R}^{n} \rightarrow \{0,1\}$  which are of the form $f(x) = \mathsf{sign}(\sum_{i=1}^n w_i x_i -\theta)$ where all of $w_1, \ldots, w_n$ and $\theta \in \mathbb{R}$.\footnote{the sign function outputs $1$ iff its argument is positive.} To understand the connection between halfspaces and stochastic knapsack,  let us consider an instance of the stochastic knapsack problem (Definition~\ref{def:knapsack}) of type $(\mathcal{D}_c, \mathbb{Z}_M)$. In other words, the items are given by $\{(\bX_\ell, v_\ell)\}$, the knapsack capacity is $C$ 
and the overflow probability is $p$ and the item sizes $\{\bX_\ell\}$ are now $(c,2,4)$ hypercontractive random variables. Now, consider any set $S \subseteq [n]$ which is feasible i.e.,~$\Pr[\sum_{\ell \in S} \bX_i >C] \le p$. This is equivalent to saying that 
the halfspace $f_S$ defined as
$f_S(\bX_1, \ldots, \bX_n) = \mathsf{sign}(\sum_{\ell \in S} \bX_\ell -C)$ is $1$ with probability at most $p$.
Ostensibly, these halfspaces  appear to be simple as all the weights $w_1, \ldots, w_n \in \{0,1\}$. However, this
simplicity is superficial as we allow $\{\bX_\ell\}$ to be arbitrary $(c,2,4)$ hypercontractive random variables and  in fact, if $\bX_\ell$ is $(c,2,4)$ hypercontractive, so is $w \cdot \bX_\ell$ for any $w \in \mathbb{R}$. 

The high level idea in the proof  of Theorem~\ref{thm:hyper-main-bounded}
is to exploit the so-called ``structure versus randomness" phenomenon for halfspaces which was introduced in the  influential work of Servedio~\cite{Servedio:07cc} and has subsequently played a crucial role in the recent developments in the complexity theoretic analysis of halfspaces~\cite{Servedio:07cc, DGJ+:10, de2014nearly, MZstoc10, MORS:10} (we explain this phenomenon a little later). Results in this line of work have \red{mostly} looked at halfspaces of the form $g(\bX_1, \ldots, \bX_n) = \mathsf{sign} (\sum_{i \in S} \bX_i)$ where $S \subseteq [n]$ and each $\bX_i$ is a so-called ``balanced Bernoulli type" random variable i.e., random variables of the form $w_i \cdot \bZ_i$ where $\bZ_i$ is a  Bernoulli random variable such that $\Pr[\bZ_i=0]$ is bounded away from $0$ and $1$ by a positive constant. In fact, most of \red{the work in complexity theory} considers the case when $\Pr[\bZ_i=0] = \Pr[\bZ_i=1]=1/2$.  Starting with the observation that balanced Bernoulli type random variables are $(O(1),2,4)$ hypercontractive, we generalize a significant fraction of the machinery from \cite{Servedio:07cc} to  arbitrary $(c,2,4)$ hypercontractive random variables.
\red{Indeed, we believe that a key conceptual contribution of this work is to realize the connection between stochastic optimization (specifically stochastic knapsack) and the ``structure versus randomness" paradigm of \cite{Servedio:07cc} for halfspaces on hypercontractive random variables. Finally, we also mention that in the context of constructing pseudorandom generators, Gopalan \emph{et al.}~\cite{GOWZ10} also extended the machinery of Servedio~\cite{Servedio:07cc} 
to hypercontractive random variables. However, they work with the stronger notion of hypercontractivity~\cite{wolff2007hypercontractivity, Krakowiak1988} alluded to earlier. While   there is some parallel between our extension of the machinery of \cite{Servedio:07cc} and that of Gopalan \emph{et al.}, it is not clear if their extension can be adapted to our setting in a black box manner.  Finally, we would like to emphasize that the main thrust of this paper  is not on Boolean function analysis but more on how it can serve as an effective tool in stochastic design problems. 

}

We now briefly explain the structure versus randomness paradigm in the context of stochastic knapsack problem as well as how it is algorithmically useful. Let us assume that  $S_{opt} = \{{j_1
}, \ldots, {j_R}\} \subseteq [n]$ is the optimal solution.  Assume that $\Var(\bX_{j_1}) \ge \ldots \ge \Var(\bX_{j_R})$. There are now two possibilities: (i) The first is that $\Var(\bX_{j_1})$ is small compared to the total sum of the variances $\sum_{\ell \in S_{opt}} 
\Var(\bX_{j_1})$. In this case, the Berry-Ess\'{e}en theorem~(Theorem~\ref{corr:BE}) implies that the distribution $\bZ_{S_{opt}} = \sum_{\ell \in S_{opt}} \bX_\ell$ (approximately) follows a Gaussian distribution. We remark that in order to get non-trivial error bounds from the Berry-Ess\'{e}en theorem, we need 
that the random variables $\{\bX_\ell\}$ are $(c,2,4)$ hypercontractive. Now, observe that a Gaussian is completely characterized by its mean and its variance (i.e., its first two moments). Using an idea similar to case (i) of Theorem~\ref{thm:Bernoulli1}, we can use dynamic programming to find another set $S \subseteq [n]$ such that $\bZ_S = \sum_{\ell \in S} \bX_\ell$ has the same first and second moments as $\bZ_{S_{opt}}$ and such that $\sum_{\ell \in S} v_\ell \ge  \sum_{\ell \in S_{opt}} v_\ell$. By applying the Berry-Ess\'{e}en theorem, we obtain that the overflow probability of $\bZ_{S}$ is at most $\epsilon$ more than that of $\bZ_{S_{opt}}$ which completes the proof of this case. This description here is significantly simplified and glosses over some key technical difficulties (which is the reason we get an $(\epsilon, \epsilon)$ approximation as opposed to an $(\epsilon,0)$ approximation in Theorem~\ref{thm:hyper-c}). 

The other possibility is if $\bX_{j_1}$ constitutes a significant fraction of the variance of $\bZ_{S_{opt}}$. In that case, the random variable $\bZ_{S_{opt},2}$ 
defined as $\sum_{\ell \in S_{opt} \setminus j_1} \bX_\ell$ has a noticeably smaller variance  than $\bZ_{S_{opt}}$ and in essence, ``we have made progress". We can now recursively look at the random variable $\bZ_{S_{opt},2}$ and apply the same argument as before. 
Intuitively, such a process can only continue for a  bounded number of steps because in each step, we ``cut-off a sizeable fraction of the variance". In particular, we show that after $L = \tilde{O}(c^4/\epsilon^2)$ such steps, the random variable $\bZ_{S_{opt}, L}$ essentially behaves like a constant.
This argument can be formalized by the notion of critical index (Definition~\ref{def:critical-index}) 
and is an extension of the eponymous notion from \cite{Servedio:07cc}.   
Roughly speaking, the critical index is   the smallest integer $K$  such that  $\bZ_{S_{opt},K}$ behaves like a Gaussian random variable. 
The reason the notion of critical index is algorithmically useful is the following. Define $T = \min \{K,L\}$. Since $T$ is upper bounded by a constant, the algorithm can guess the indices $\{j_1, \ldots, j_T\}$. On the other hand, the random variable $\bZ_{S_{opt},T}$ either behaves like a constant (if $K \ge L$) or like a Gaussian random variable (if $K<L$). Both these cases can be handled using dynamic programming techniques discussed earlier.

 \red{At a thematic level, the  strategy follows 
the usual ``critical-index" machinery of \cite{Servedio:07cc}. However,  simultaneously extending this machinery to arbitrary $(c,2,4)$ hypercontractive random variables as well as adapting 
it in the context of stochastic knapsack poses several challenges (which are difficult to explain at this level of detail). Also, we introduce some new technical tools such as the Kolmogorov-Rogozin inequality~(Lemma~\ref{lem:Kolmogorov}) etc. which do not seem to have been explicitly used before in this line of work and can potentially be useful elsewhere. }


  Finally, we mention that the critical-index machinery was also used by Daskalakis \emph{et al.}~\cite{DDDMS:14} in the context of stochastic optimization; in particular, to obtain an approximation scheme for so-called fault tolerant distributed storage. Very briefly, given balanced Bernoulli random variables\footnote{i.e., the probability of being $0$ is bounded away from both $0$ and $1$.} $\bY_1, \ldots, \bY_n$ and a threshold $C$, they seek to find a vector $w \in [0,1]^n$ such that (i) $\sum_{i=1}^n w_i=1$ and (ii) $\Pr[\sum_{i=1}^n w_i \cdot \bY_i > C]$ is maximized. 
\red{While there is some ostensible similarity between their problem and ours, there are fundamental differences: namely, their solution space is the $n$-dimensional polytope and indeed, a significant use of the critical machinery in \cite{DDDMS:14} is to argue that there is an approximately optimal solution with a ``nice", so-called  ``anti-concentrated" solution. In contrast, our solution space is combinatorial (namely subsets of $[n]$) and we use the critical index machinery to characterize the probabilistic behavior of $\sum_{S \subseteq [n]} \bX_S$ for all $S \subseteq [n]$. 
Finally, we also emphasize that \cite{DDDMS:14} only dealt with 
sums of balanced Bernoulli type random variables whereas we have to tackle sums of independent $(c,2,4)$ hypercontractive random variables thus creating additional complications. 

To sum up, a wealth of sophisticated and powerful results have been developed in probability theory and the complexity theoretic study of halfspaces that have direct relevance to the linear forms in random variables that are at the heart of the stochastic knapsack problem.  We view the transfer of these ideas and techniques from complexity theory and probability to stochastic optimization as a conceptual contribution of this work, and we hope that more connections will be uncovered between these previously rather disjoint fields.

}

\section{Some basics of probability theory}


In this section, we list some probabilistic preliminaries which will be useful throughout the paper.


\subsubsection*{Distance between distributions}
We will use two (well-known) notions of distances between real-valued random variables which we recall below. 
\begin{definition}
For real-valued random variables $\bX$ and $\bY$, 
$$
\dcdf(\bX,\bY) = \sup_{t \in \mathbb{R}} \big|\Pr[\bX \le t]- \Pr[\bY \le t]\big|, 
$$
$$
\dtv(\bX,\bY) = \sup_{A \subseteq \mathbb{R}} \big| \Pr[\bX \in A] - \Pr[\bY \in A] \big|.
$$
Here the supremum $A$ is taken over any measurable subset of $\mathbb{R}$. It is easy to see that $\dtv(\bX,\bY)$ (up to a factor of $2$) is the same as the $\ell_1$ distance between the random variables $\bX$ and $\bY$. 
\end{definition}

\subsubsection{Anti-concentration and smoothness of random variables}

The notion of anti-concentration of random variables is going to play an important role in the proof of Theorem~\ref{thm:hyper-c}. We quantify the notion of anti-concentration of a real-valued random variable by the so-called L\'{e}vy concentration function (defined below).
\begin{definition}~\label{def:Levy}
For a real-valued random variable $\bX$ and $t >0$, we define $Q_{\bX}(t)$ as
$Q_{\bX}(t) = \sup_{a \in \mathbb{R}} \Pr[a \le {\bX} \le a+t]$. 
\end{definition}
Note that $Q_{\bX}(t)$ is  an upper bound on the mass
that $\bX$ puts in any interval of size $t$. A useful intuition for $Q_{\bX}(t)$ is that it is a measure of smoothness of the random variable $t$. 
We now record a very simple but useful fact about the function $Q_{\bX}(t)$, namely that it decreases upon convolution.
\begin{fact}~\label{fact:convolve}
Let ${\bX}$ and $\bY$ be independent random variables. Then, for $t  >0$, $Q_{\bX+\bY}(t) \le Q_{\bX}(t)$.
\end{fact}
The next lemma shows that hypercontractive random variables have non-trivial bounds on $Q_X( \cdot )$. 
\begin{lemma}~\label{lem:hypercontractive}
Let  $\bX$ be a  $(c,2,4)$-hypercontractive random variable with $\mu_2(\bX) = \sigma^2$. Then, for $t = \sigma/2$ and $\delta = \frac{9}{128 \cdot (c+2)^4}
$, $Q_{\bX}(t) \le 1-\delta$. \end{lemma}
\begin{proof}
We begin with a simplification. Namely, let $\bZ = \bX-\bX'$ where $\bX'$ is an i.i.d. copy of $\bX$. Note that $\mathbf{E}[\bZ]=0$, $\mu_2(\bZ)  = 2 \mu_2(\bX)$ and $\mu_4(\bZ) = 2 \mu_4(\bX)  + 6 \mu_2^2(\bX)$.  It easily follows $\bZ$ is $(c+2, 2,4)$-hypercontractive.  Now, towards a contradiction assume that $Q_{\bX}(t) > 1-\delta$. 
Then, it follows that $\Pr[|\bZ| \le t] > 1-2\delta$.   
Let us define $\kappa = \Pr[|\bZ|>t]$. Also, 
let $\bY$ be the random variable obtained by conditioning $\bZ$ on the event $|\bZ|>t$. Also, let $p(\cdot)$ be the density of $\bZ$. 
\begin{eqnarray*}
\sigma^2 = \int_{x \in \mathbb{R}} x^2 p(x) dx  &=& \int_{|x|\le t} x^2 p(x) dx + \int_{|x| >t} x^2 p(x) dx  \\
&\leq& t^2 \cdot (1-\kappa) +  \int_{|x| >t} x^2 p(x) dx \\
&=& t^2 \cdot (1-\kappa) +  \kappa \cdot \mathbf{E}[Y^2]  
\end{eqnarray*} 
As a consequence, 
we have \begin{equation}\label{eq:by-1}
\mathbf{E}[\bY^2] \ge \frac{1}{\kappa} \cdot (\sigma^2 - t^2(1-\kappa)).
\end{equation}


Likewise, it is clear that
\begin{equation}~\label{eq:fourth-1}
\mathbf{E}[\bY^4] \le \frac{1}{\kappa} \cdot \mathbf{E}[\bZ^4] \le \frac{(c+2)^4}{\kappa} \cdot (\mathbf{E}[\bZ^2])^2 =\frac{4 \cdot (c+2)^4}{\kappa} \cdot  \sigma^4
\end{equation}
Applying Jensen's inequality on (\ref{eq:by-1}) and (\ref{eq:fourth-1}), we get 
$$
\frac{4 \cdot (c+2)^4}{\kappa} \cdot  \sigma^4 \ge \frac{1}{\kappa^2} \cdot (\sigma^2 - t^2(1-\kappa))^2
$$
Plugging in $t = \sigma/2$, we get 
$$
\kappa \ge \frac{9}{64 \cdot (c+2)^4} \ \textrm{and} \ \delta  \ge \frac{9}{128 \cdot (c+2)^4}.
$$

\end{proof}
The following well-known inequality, known as the Kolmogorov-Rogozin inequality~\cite{Kol:58aip, Rog:61tpa} states that adding independent random variables improves anti-concentration. 
\begin{lemma}[Kolmogorov-Rogozin inequality]~\label{lem:Kolmogorov}
Let $\bX_1, \bX_2, \ldots, \bX_n$ be independent random variables  and let $\bZ = \bX_1 + \ldots + \bX_n$. Then, for $t>0$ and $0<t_i\le t$ (for $i=1, \ldots, n$), we have
\[
Q_{\bZ}(t) \le \frac{100 \cdot t}{\sqrt{\sum_{i=1}^n  t_i^2 \cdot (1-Q_{\bX_i}(t_i))}}
\]
\end{lemma}
\subsubsection{Berry-Ess\'{e}en theorem and other central limit theorems}
Quantitative versions of the central limit theorem will be a key ingredient in nearly all the theorems. We begin with the Berry-Ess\'{e}en theorem~\cite{Feller} which implies convergence in cdf distance. Let $\mathcal{N}(\mu, \sigma^2)$ denote the Gaussian with mean $\mu$ and variance $\sigma^2$. 
\begin{theorem}~\label{thm:BE}\emph{(Berry-Ess\'{e}en theorem)} Let $\bX_1, \ldots, \bX_n$ be independent random variables and let $\bZ= \bX_1  + \ldots + \bX_n$, $\mu =\mathbf{E}[\bZ]$ and $\sigma^2= \mathsf{Var}(\bZ)$. Then, 
$$
\dcdf (\bZ, \mathcal{N}(\mathbf{E}[\bZ], \mathsf{Var}(\bZ)) \leq \frac{1}{\sigma} \cdot \max_{1 \le i \le n} \frac{\mu_3(\bX_i)}{\mu_2(\bX_i)}.
$$
Note that by Lyupanov's inequality~\cite{Garling}, for any random variable $X$, $\mu_3(\bX) \le \sqrt{\mu_2(\bX) \cdot \mu_4(\bX)}$. Thus, 
$$
\dcdf (\bZ, \mathcal{N}(\mathbf{E}[\bZ], \mathsf{Var}(\bZ)) \leq \frac{1}{\sigma} \cdot \max_{1 \le i \le n} \frac{\sqrt{\mu_4(\bX_i)}}{\sqrt{\mu_2(\bX_i)}}. 
$$
\end{theorem}
As a consequence, we have the following corollary which is applicable to $(c,2,4)$ hypercontractive random variables. 
\begin{corollary}~\label{corr:BE}
Let $\bX_1, \ldots, \bX_n$ be independent $(c,2,4)$ hypercontractive random variables and let $\bZ = \bX_1 + \ldots + \bX_n$, $\mu = \mathbf{E}[\bZ]$ and $\sigma^2  = \mathsf{Var}(\bZ)$. Then, 
$$
\dcdf (\bZ, \mathcal{N}(\mathbf{E}[\bZ], \mathsf{Var}(\bZ)) \leq \max_i \frac{c^2 \cdot \sqrt{\mu_2(\bX_i)}}{\sigma}.
$$
\end{corollary}
The next limit theorems are applicable only to sums of Bernoulli random variables but provide a stronger convergence guarantee, namely in total variation distance. The translated Poisson distribution $\mathsf{TP}(\mu, \sigma^2)$ is the Poisson distribution 
$\mathsf{Poi}(\lambda)$ translated by $\lfloor \mu  - \sigma^2 \rfloor$ and $\lambda = \mu - \lfloor \mu  - \sigma^2 \rfloor$. Note that $\mathsf{Poi}(\lambda)$ is the Poisson with mean $\lambda$.

\begin{lemma}~\label{lem:CLT1}~\emph{Translated Poisson approximation~\cite{Rollin:07}} Let $\bX_1, \ldots, \bX_n$ be independent Bernoulli random variables and let $\bZ = \sum_{i=1}^n \bX_i$. Let $\mu = \mathbf{E}[\bZ]$ and $\sigma^2 =\mathsf{Var}(\bZ)$. Then, 
\[
\dtv\big( \bZ, \mathsf{TP}(\mu, \sigma^2)) \le \frac{\sigma +2}{\sigma^2}. 
\]
\end{lemma}
\begin{lemma}~\label{lem:CLT2}~\emph{Poisson approximation~\cite{BHJ:92}} Let $\bX_1, \ldots, \bX_n$ be independent Bernoulli random variables and for $1 \le j \le n$, 
$\mathbf{E}[\bX_j] =p_j$. For $\bZ = \sum_{i=1}^n \bX_i$ and $\mu = \sum_{j=1}^n p_i$, 
$$
\dtv \big(\bZ, \mathsf{Poi}(\mu) \big) \le \frac{\sum_{i=1}^n p_i^2}{\sum_{i=1}^n p_i}
$$

\end{lemma}
\begin{lemma}~\label{lem:moment-matching}~\emph{Moment matching theorem for PBDs~\cite{Roos:00, Daskalakis2015}} 
Let $\{\bX_i\}_{i=1}^n$ and  $\{\bY_i\}_{i=1}^n$ be two families of independent Bernoulli random variables such that for all $1 \le i \le n$, $\mathbf{E}[\bX_i]=p_i, \ \mathbf{E}[\bY_i]=q_i \le 1/2$. Let $\bZ_{\bX} = \sum_{i=1}^n \bX_i$, $\bZ_{\bY}  = \sum_{i=1}^n \bY_i$ and for $1 \le j \le T$,  $\sum_{i \in [n]} p_i^j = \sum_{i \in [n]} q_i^j$. 
Then, $\dtv(\bZ_{\bX}, \bZ_{\bY})$ is bounded by 
$$
\dtv(\bZ_{\bX}, \bZ_{\bY}) \leq 13 \cdot  (T+1)^{1/4} \cdot 2^{-(T+1)/2}. 
$$
For sufficiently large $T$, the right hand side is upper bounded by $2^{-T/3}$. 
\end{lemma}

We next define Poisson multinomial distributions and state a moment matching theorem for these distributions. 
\begin{definition}
A random variable $\bX$ supported on $\{e_1, \ldots, e_k\}$ (where $e_i$ is the standard unit vector in the $i^{th}$ direction) is said to be a $k$-categorical random variable (CRV). A $(n,k)$ Poisson-multinomial distribution is obtained by adding $n$ independent random variables $\bX_1, \ldots, \bX_n$ where each $\bX_i$ is a $k$-CRV. 
\end{definition}

The following moment matching theorem was proven by Daskalakis, Kamath and Tzamos~\cite{daskalakis2015structure}.  To do this, for  positive integers $w$ and $k$, define $V_k(w) = \{ v \in \mathbb{Z}^k: v_i \ge 0 \ \wedge \sum_{i} v_i \le w\}$. 
\begin{theorem}~\label{thm:moment-matching-PMD}
Let $\{\bX_1, \ldots, \bX_n\}$  and $\{\bY_1, \ldots, \bY_{m}\}$ be independent $k$-CRVs. They satisfy the following properties:
\begin{itemize}
\item For all $1 \le j \le k$, $|\max_i \Pr[\bX_i=e_j] - \min_i \Pr[\bX_i=e_j]| \le \frac{1}{4ek^3}$ and $|\max_i \Pr[\bY_i=e_j] - \min_i \Pr[\bY_i=e_j]| \le \frac{1}{4ek^3}$. 
\item There exists $j_0 \in [k]$ such that $\sum_{i=1}^n  \Pr[\bX_i=e_{j_0}] \ge \frac{n}{k}$ and  $\sum_{i=1}^m  \Pr[\bY_i=e_{j_0}] \ge \frac{m}{k}$. 
\item For all $\alpha \in V_k(w)$, let
$$
\sum_{i=1}^n \prod_{j=1}^k \Pr[\bX_i=e_j]^{\alpha_j} = \sum_{i=1}^m \prod_{j=1}^k \Pr[\bY_i=e_j]^{\alpha_j}. 
$$
\end{itemize}
Then, $\dtv(\sum_{i=1}^n \bX_i, \sum_{i=1}^m \bY_i)  \le 2^{-w+1}$. 
\end{theorem}
An easy corollary of this is the following. 
\begin{corollary}~\label{corr:PMD-moment-matching}
Let $A = \{a_1, \ldots, a_k\}$ and let $\{\bX_1, \ldots, \bX_n\}$ and $\{\bY_1, \ldots, \bY_{m}\}$ be independent random variables supported on $A$. Assume that they satisfy: 
\begin{itemize}
\item For all $1 \le j \le k$, $|\max_i \Pr[\bX_i=a_j] - \min_i \Pr[\bX_i=a_j]| \le \frac{1}{4ek^3}$ and $|\max_i \Pr[\bY_i=a_j] - \min_i \Pr[\bY_i=a_j]| \le \frac{1}{4ek^3}$. 
\item There exists $j_0 \in [k]$ such that for all $1 \le i \le n$, $\Pr[\bX_i = a_{j_0}] = \max_{j \in [k]} \Pr[\bX_i = a_j]$ and $1 \le i \le m$, $\Pr[\bY_i = a_{j_0}] = \max_{j \in [k]} \Pr[\bY_i = a_j]$. 
\item For all $\alpha \in V_k(w)$, let
$$
\sum_{i=1}^n \prod_{j=1}^k \Pr[\bX_i=a_j]^{\alpha_j} = \sum_{i=1}^m \prod_{j=1}^k \Pr[\bY_i=a_j]^{\alpha_j}. 
$$
\end{itemize}
Then, $\dtv(\sum_{i=1}^n \bX_i, \sum_{i=1}^m \bY_i)  \le 2^{-w+1}$.
\end{corollary}
\begin{proof}
Corresponding to each $\bX_i$, define the $k$-CRV $\tilde{\bX}_i$ as follows: For every $1 \le i \le n$ and $1 \le j \le k$, $\Pr[\tilde{\bX}_i = e_j] = \Pr[\bX_i = a_j]$. Likewise, for each $\bY_i$, define the $k$-CRV 
$\tilde{\bY}_i$ as $\Pr[\tilde{\bY}_i = e_j] = \Pr[\bY_i = a_j]$. Note that the three conditions of Corollary~\ref{corr:PMD-moment-matching} imply the three conditions required to apply Theorem~\ref{thm:moment-matching-PMD} for $\{\tilde{\bX}_i\}$ and $\{\tilde{\bY}_i\}$.
Applying Theorem~\ref{thm:moment-matching-PMD}, we obtain $\dtv(\sum_{i=1}^n \tilde{\bX}_i, \sum_{i=1}^m \tilde{\bY}_i) \le 2^{-w+1}$. Finally, note that $\sum_{i=1}^n \bX_i= \langle \overline{a}, \sum_{i=1}^n \tilde{\bX}_i \rangle$ and $\sum_{i=1}^n \bY_i= \langle \overline{a}, \sum_{i=1}^n \tilde{\bY}_i \rangle$ where $\overline{a} = (a_1, \ldots, a_k)$. This proves the corollary. 
\end{proof}

\section{Proof of Theorem~\ref{thm:Bernoulli1}}
We first make the following simple observation (which will be useful in the proof of Theorem~\ref{thm:hyper-main-bounded} as well).  
\begin{proposition}~\label{prop:verify}
Given (efficiently samplable)  random variables $\bX_1, \ldots, \bX_n$, capacity $C$, a subset $S \subseteq [n]$ and an error parameter $\epsilon>0$, there is a randomized $\mathsf{poly}(n/\epsilon)$ time algorithm which computes $\Pr[\sum_{j \in S} \bX_i>C]$ to additive error $\pm \epsilon$. 
\end{proposition} 
\begin{proof}
This is the consequence of a trivial sampling algorithm. \end{proof}
We will use this proposition in a simple way. Namely, for any $S \subseteq [n]$, we use the notation 
$\Pr[\sum_{j \in S} \bX_i>C]\le_\epsilon q$ to denote that a $\pm \epsilon$ additive approximation to $\Pr[\sum_{j \in S} \bX_i>C]$ is bounded by $q$. Here are a few key observations about this relation. 
\begin{enumerate}
\item[(i)] If $\Pr[\sum_{j \in S} \bX_i>C]\le_\epsilon q$, then $\Pr[\sum_{j \in S} \bX_i>C] \le q+\epsilon$. 
\item[(ii)] If $\Pr[\sum_{j \in S} \bX_i>C] \le q-\epsilon$, then  $\Pr[\sum_{j \in S} \bX_i>C]\le_\epsilon q$. 
\item[(iii)] There is a randomized algorithm to check this relation in polynomial time. While the randomized algorithm has a non-zero probability of failure, it can be made inverse exponentially small in $n$ by increasing the running time by a factor of $O(n)$. Thus, for simplicity, we will assume that the relation is computed with probability $1$. 
\end{enumerate}
We will divide the proof of Theorem~\ref{thm:Bernoulli1} into two claims. First of all, given any stochastic knapsack instance  $\{(\bX_i,v_i)\}_{i=1}^n$ of type $(\mathcal{D}_B, \mathbb{Q}^+)$, capacity $C$, overflow probability $p$, error parameter $\epsilon>0$ and profit value $V$, define $\feas_{p,V}$ as 
$$
\feas_{p,V} = \{S \subseteq [n]: \Pr[\sum_{j \in S}\bX_j>C]\le p  \ \textrm{and} \ \sum_{j \in S} v_j =V\}.
$$
Let $V_{opt}$ be the maximum $V$ such that $\feas_{p,V}$ is non-empty. The algorithm in Theorem~\ref{thm:Bernoulli1} is a combination of two algorithms: The first one succeeds if $\mathsf{Var}(\sum_{j \in S} \bX_j)$ is large where $S$ is the target set in $\feas_{p,V_{opt}}$ and the second one succeeds if $(\sum_{j \in S} \bX_j)$ is small.  
 {Also, from now on, we will assume that $\epsilon>0$ is smaller than any explicitly specified constant occuring in our proofs. }
\begin{claim}~\label{clm:SK-large}
There is an algorithm \textsf{SK-Bernoulli-Large} with the following guarantee: Given a stochastic knapsack instance $\{(\bX_i,v_i)\}_{i=1}^n$ of type $(\mathcal{D}_B, \mathbb{Z}_M^+)$, capacity $C$, overflow probability $p$, error parameter $\epsilon>0$, it outputs a set $S^{\ast}$ with the following guarantee: For a 
profit value $V$, define the set $\feas_{p,V,1}$ as
$$
\feas_{p,V,1} = \{S \subseteq [n]: \Pr[\sum_{j \in S}\bX_j>C]\le p, \ \mathsf{Var}(\sum_{j \in S} \bX_j) \ge 1/\epsilon^2 \ \textrm{and} \ \sum_{j \in S} v_j =V\}. 
$$
Let $V_{opt,1}$ be the maximum value such that $\feas_{p,V,1}$ is non-empty. Then, $S^\ast \in \feas_{p+4 \cdot \epsilon}$ and $\sum_{j \in S^\ast}v_j \ge V_{opt,1}$. The running time of the algorithm is $\mathsf{poly}(n,1/\epsilon)$. 
\end{claim}
\begin{claim}~\label{clm:SK-small}
There is an algorithm \textsf{SK-Bernoulli-Small} with the following guarantee: Given a stochastic knapsack instance $\{(\bX_i,v_i)\}_{i=1}^n$ of type $(\mathcal{D}_B, \mathbb{Z}_M^+)$, capacity $C$, overflow probability $p$, error parameter $\epsilon>0$, it outputs a set $S^{\ast}$ with the following guarantee: For a profit 
$V$, define the set $\feas_{p,V,2}$ as
$$
\feas_{p,V,2} = \{S \subseteq [n]: \Pr[\sum_{j \in S}\bX_j>C]\le p, \ \mathsf{Var}(\sum_{j \in S} \bX_j) \le 1/\epsilon^2 \ \textrm{and} \ \sum_{j \in S} v_j =V\}. 
$$
Let $V_{opt,2}$ be the maximum value such that $\feas_{p,V,2}$ is non-empty. Then, $S^\ast \in \feas_{p+4 \cdot \epsilon}$ and $\sum_{j \in S^\ast}v_j \ge V_{opt,2}$. The running time of the algorithm is $\mathsf{poly}(n,(1/\epsilon)^{\log^2(1/\epsilon)})$. 
\end{claim}
Note that Theorem~\ref{thm:Bernoulli1} follows easily as a combination of Claim~\ref{clm:SK-large} and Claim~\ref{clm:SK-small}.  Let $V_{opt}$ be the maximum value for which $\feas_{p,V}$ is non-empty.  For $C_0=8$, run \textsf{SK-Bernoulli-Large} and \textsf{SK-Bernoulli-Small} with error parameter $\epsilon/C_0$. 
Let the output sets be $S^{\ast,\ell}$ and $S^{\ast, s}$ respectively. We discard $S \in  \{S^{\ast,\ell}, S^{\ast, s}\}$ if $\Pr[\sum_{j \in S} \bX_j >C] \le_{\epsilon/4
} p +\frac{3\epsilon}{4}$. We are guaranteed that both 
$S^{\ast,\ell}$ and $S^{\ast, s}$ are not discarded. We now output the set $S \in \{S^{\ast,\ell}, S^{\ast, s}\}$ which maximizes $\sum_{j \in S} v_j$. 
This finishes the proof of Theorem~\ref{thm:Bernoulli1}.

\begin{proofof}{Claim~\ref{clm:SK-large}}
For $\ell \in [n]$, let us define $p_\ell = \mathbf{E}[\bX_\ell]$ and $q_\ell$ be $p_\ell$ rounded to the nearest multiple of $\epsilon/(4n)$. For $ \ell \in [n]$, let $\{\bY_\ell\}_{\ell=1}^n$ be independent Bernoulli random variables such that $\mathbf{E}[\bY_\ell]=q_\ell$.
Define the set $A_1$ and $A_2$ as: 
\[
A_1 = \bigg\{ j \cdot \frac{\epsilon}{4n} : j \in \mathbb{N} \ \textrm{and} \ j \cdot \frac{\epsilon}{4n} \le n \bigg\}, \ 
A_2 = \bigg\{ j \cdot \frac{\epsilon^2}{16n^2} : j \in \mathbb{N} \ \textrm{and} \ j \cdot \frac{\epsilon^2}{16n^2} \le n \bigg\}. 
\]
 For $1 \le \ell \le n$, define item $J_\ell$ with size  $(q_\ell, q_\ell^2)$ and profit  $v_\ell$. We are now ready to define the algorithm \textsf{SK-Bernoulli-Large}. 
\begin{enumerate}
\item Let $V_{\max} =0$ and $S_{\max} =\phi$. 
\item For $(x,y) \in A_1 \times A_2$, 
\item \hspace{10pt} Run \textsf{Pseudo-knapsack} with items $\{J_\ell\}_{\ell=1}^n$, target size $(x,y)$ and quantization is $(\epsilon/4n, \epsilon^2/16n^2)$. 
\item  \hspace{10pt} If the output is $S \subseteq [n]$ and $\Pr[\sum_{\ell \in S} \bX_j >C] \le_{\epsilon/2} p+3.5 \epsilon$ and $\sum_{\ell \in S} v_\ell >V_{\max}$, $S_{\max} \leftarrow S$. 
\item Output $S_{\max}$. 
\end{enumerate}
The running time is computed as follows: Every invocation of \textsf{Pseudo-knapsack} takes time $\mathsf{poly}(n,1/\epsilon)$. Since the cardinality of $A_1 \times A_2$ is $\mathsf{poly}(n/\epsilon)$, the total running time is 
 $\mathsf{poly}(n,1/\epsilon)$. 
 
 To prove correctness, it suffices to show that there exists $(x_0, y_0) \in A_1 \times A_2$ such that the output of \textsf{Pseudo-knapsack} with target $(x_0,y_0)$  returns set $S$ such that {$\Pr[\sum_{j \in S} \bX_j>C] \le p+3\epsilon$ } and
$\sum_{j \in S} v_j \ge V_{opt,1}$.  
To show this, let $V = V_{opt,1}$ and let $S_{opt} \subseteq [n]$ be such that $S_{opt} \in \feas_{p,V,1}$ and $ V = \sum_{\ell \in  S_{opt}}v_\ell$. Then, it follows that $\Pr[\sum_{\ell \in S_{opt}} \bY_\ell >C] \le p+\epsilon/4$. Now, note that by construction $\sum_{\ell \in S_{opt}} q_\ell \in A_1$ and $\sum_{\ell \in S_{opt}} q_\ell^2 \in A_2$. Let $x_0=\sum_{\ell \in S_{opt}} q_\ell$ and $y_0=\sum_{\ell \in S_{opt}} q_\ell^2$. Then, the routine \textsf{Pseudo-knapsack} returns set $S$ such that 
$\sum_{\ell \in S} q_\ell = x_0$, $\sum_{\ell \in S} q_\ell^2 = y_0$ and $\sum_{\ell \in S} v_\ell \ge V$. 
This implies that 
\begin{equation}~\label{eq:a11}
\mathbf{E}\big[\sum_{\ell \in S} \bY_\ell \big] = \sum_{\ell \in S} q_\ell = \sum_{\ell \in S_{opt}} q_\ell = \mathbf{E}\big[\sum_{\ell \in S_{opt}} \bY_\ell \big] . 
\end{equation}
\begin{equation}~\label{eq:a21}
\mathsf{Var}\bigg(\sum_{\ell \in S} \bY_\ell\bigg)  = \sum_{\ell \in S} q_\ell - q_\ell^2 =\sum_{\ell \in S_{opt}} q_\ell - q_\ell^2 =\mathsf{Var}\bigg(\sum_{\ell \in S_{opt}} \bY_\ell\bigg)
\end{equation}
Further, note that 
\[
\mathsf{Var}\bigg(\sum_{\ell \in S_{opt}} \bY_\ell\bigg)=\mathsf{Var}\bigg(\sum_{\ell \in S} \bY_\ell\bigg)  \ge \mathsf{Var}\bigg(\sum_{\ell \in S} \bX_\ell\bigg) - \frac{\epsilon}{4n} \cdot n \ge \frac{1}{\epsilon^2} - \epsilon \ge \frac{0.9}{ \epsilon^2}. 
\]
The last inequality relies on assuming that $\epsilon>0$ is sufficiently small. Combining Lemma~\ref{lem:CLT1} and (\ref{eq:a11}), (\ref{eq:a21}), we obtain 
\[
\dtv \big(\sum_{\ell \in S_{opt}} \bY_\ell,\sum_{\ell \in S} \bY_\ell \big) \le \frac{2\epsilon}{0.9}.
\]
Now, observe that
\[
\dtv\big(\sum_{\ell \in S_{opt}} \bY_\ell,\sum_{\ell \in S_{opt}} \bX_\ell \big) \le \frac{\epsilon}{4}, \quad \dtv\big(\sum_{\ell \in S} \bY_\ell,\sum_{\ell \in S} \bX_\ell \big) \le \frac{\epsilon}{4}
\]
This implies that 
\[
\dtv \big( \sum_{\ell \in S} \bX_\ell, \sum_{\ell \in S_{opt}} \bX_\ell\big) < 3 \epsilon.  
\]
As $\Pr[\sum_{j \in S_{opt}}\bX_j>C] \le p$, we have $\Pr[\sum_{j \in S_{}}\bX_j>C] \le p + 3 \epsilon$. This concludes the proof.  
\end{proofof}

\begin{proofof}{Claim~\ref{clm:SK-small}}
Let us begin by defining a partition of $[n]$ into four sets $B_1, B_2, B_3, B_4$ as follows: 
$$
B_1 = \big\{\ell \in [n]: p_\ell \le \frac{\epsilon}{100}\big\}, \quad B_2 = \big\{\ell \in [n]: p_\ell \ge 1- \frac{\epsilon}{100}\}, 
$$
$$
B_3 = \big\{\ell \in [n]: \frac{\epsilon}{100} \le p_\ell \le 1/2\big\}, \quad 
B_4 = \big\{\ell \in [n]: 1/2 \le p_\ell \le 1- \frac{\epsilon}{100}\big\}. 
$$
For $1 \le \ell \le n$, we define $q_\ell $ as follows: For elements in $B_1$ and $B_2$, $q_\ell$ is $p_\ell$ rounded to the nearest multiple of $\frac{\epsilon}{4n}$. For elements in $B_3$ and $B_4$, $q_\ell$ is $p_\ell$ rounded to the nearest multiple of $\epsilon^4/1000$.
First, let us define the set $A_1$ (similar to the proof of Claim~\ref{clm:SK-large}) as
\[
A_1 = \bigg\{ j \cdot \frac{\epsilon}{4n} : j \in \mathbb{N} \ \textrm{and} \ j \cdot \frac{\epsilon}{4n} \le n \bigg\}. 
\] 
Next, let us define $T_0 = 4 \log (1/\epsilon)$ and for $T \in \{1,2, \ldots, T_0\}$, we define the set $\widetilde{A}_T$  as
\[
\widetilde{A}_T = \bigg\{ j \cdot \frac{\epsilon^{4T}}{1000^T}: j \in \mathbb{N} \ \textrm{and} \ j \cdot \frac{\epsilon^{4T}}{1000^T} \le \frac{1}{\epsilon^2}\bigg\}.
\]
Next, for $\ell \in [n]$, we construct items $J_\ell$ whose sizes are defined as follows. 
$$
J_\ell = \begin{cases} (q_\ell) &\textrm{if} \ \ell \in B_1 
\\(1,1-q_\ell) &\textrm{if} \ \ell \in B_2 \\
(q_\ell, q_\ell^2, \ldots, q_\ell^{T_0}) &\textrm{if} \ \ell \in B_3 \\
(1,(1-q_\ell), (1-q_\ell)^2, \ldots, (1-q_\ell)^{T_0})&\textrm{if} \ \ell \in B_4 \\
\end{cases}
$$
Further, for all $\ell \in [n]$, the profit of $J_\ell$ is defined to be $v_\ell$. 
 We are now ready to describe \text{SK-Bernoulli-Small}. 
\begin{enumerate}
\item Set $V_{\max} =0$ and $S_{\max} = \phi$. 
\item  For $x_1 \in A_1$, $x_2 \in [n] \times A_1$, $x_3 \in \tilde{A}_1 \times \ldots \times \tilde{A}_{T_0}$ and $x_4 \in [n] \times \tilde{A}_1 \times \ldots \times \tilde{A}_{T_0}$, 
\item \hspace*{5pt}Run \textsf{Pseudo-knapsack} with items $\{J_\ell\}_{\ell \in B_1}$, target $x_1$ and quantization $(\epsilon/(4n))$. 
\item  \hspace*{5pt}Run \textsf{Pseudo-knapsack} with items $\{J_\ell\}_{\ell \in B_2}$, target $x_2$ and quantization $(1,\epsilon/(4n))$. 
\item \hspace*{5pt} Run \textsf{Pseudo-knapsack} with items $\{J_\ell\}_{\ell \in B_3}$, target $x_3$ and quantization $(\epsilon^4/1000,\ldots, \epsilon^{4T_0}/1000^{T_0})$. 
\item \hspace*{5pt} Run \textsf{Pseudo-knapsack} with items $\{J_\ell\}_{\ell \in B_4}$, target $x_4$ and quantization $(1,\epsilon^4/1000,\ldots, \epsilon^{4T_0}/1000^{T_0})$. 
\item \hspace*{5pt} Let the outputs of Steps 3, 4, 5, 6 be $S_1, S_2, S_3, S_4$. Let $S = S_1 \cup S_2 \cup S_3 \cup S_4$.
\item  \hspace*{5pt} If $\Pr[\sum_{j \in S} \bX_j >C] \le_{\epsilon/2} p +3.5 \epsilon$ and $\sum_{j \in S} v_j \ge V_{\max}$, set $S_{\max} \leftarrow S$ and $V_{\max} =   \sum_{j \in S} v_j$. 
\item Output the set $S_{\max}$. 
\end{enumerate}
Note that the total number of choices for  $(x_1, x_2)$ is $\mathsf{poly}(n/\epsilon)$,  $(x_3,x_4)$ is $n \cdot (1/\epsilon)^{O(T_0^2)}$. Further, for a fixed choice of $(x_1, x_2, x_3, x_4)$, Theorem~\ref{thm:multidimensional} implies that the running time of Steps 3, 4, 5, 6 is upper bounded by $\mathsf{poly}(n, (1/\epsilon)^{T_0^2})$. As $T_0 = O(\log(1/\epsilon))$, this implies our upper bound on the running time.

As in Claim~\ref{clm:SK-large}, it suffices to show that if for $V = V_{opt,2}$, $\feas_{p,V,2}$ is non-empty, 
then there exists $(x_1, x_2, x_3, x_4)$ such that if the algorithm returns $S_1, S_2, S_3, S_4$, then $S= S_1 \cup S_2 \cup S_3 \cup S_4$,  $\Pr[\sum_{j \in S} \bX_j >C] \le p + 3 \epsilon$ and $\sum_{j \in S} v_j \ge V$. To show this, let $S_{opt} \in \feas_{p,V,2}$. For $1 \le i \le 4$, $S_{opt,i} = S_{opt} \cap B_i$ and $V_{opt,2,i} = \sum_{\ell \in S_{opt,i}} v_\ell$. Let us define $x_i= \sum_{\ell \in S_{opt,i}} J_\ell$. Then, steps 3-6 of the algorithm return sets $S_1, S_2, S_3, S_4$ such that for all $1 \le i \le 4$, $\sum_{\ell \in S_i} J_\ell = x_i$ and $\sum_{\ell \in S_i} v_\ell \ge V_{opt,2,i}$. We now claim that  for all $1 \le i \le 4$, \begin{equation}~\label{eq:diff-3}\dtv(\sum_{\ell \in S_i} \bY_\ell , \sum_{\ell \in S_{opt,i}} \bY_\ell) \le \epsilon/4 \end{equation}
\\ \textbf{Case $i=1$:} We apply Lemma~\ref{lem:CLT2} to obtain 
\[
\dtv\bigg(\sum_{\ell \in S_1} \bY_\ell, \mathsf{Poi}\bigg(\sum_{\ell \in S_1} \mathbf{E}[\bY_\ell]\bigg) \bigg) \le \max_{\ell \in S_1} q_\ell \ \textrm{and} \ \dtv\bigg(\sum_{\ell \in S_{opt,1}} \bY_\ell, \mathsf{Poi}\bigg(\sum_{\ell \in S_{opt,1}} \mathbf{E}[\bY_\ell]\bigg)\bigg) \le \max_{\ell \in S_{opt,1}} q_\ell
\]
As $q_\ell$ is obtained by rounding $p_\ell$ to the nearest multiple of $\epsilon/(4n)$, hence applying the definition of $B_1$, we obtain that $\max_{\ell \in B_1} q_\ell \le \epsilon/100$. Additionally, by guarantee of \textsf{Pseudo-knapsack} (in Step 3 of the algorithm), we have $\sum_{\ell \in S_{opt,1}} \mathbf{E}[\bY_\ell]= \sum_{\ell \in S_{1}} \mathbf{E}[\bY_\ell]$. This implies
$$
\dtv(\sum_{\ell \in S_{opt,1}} \bY_\ell, \sum_{\ell \in S_{1}} \bY_\ell) \le \frac{\epsilon}{50}, 
$$
thus proving (\ref{eq:diff-3}) for $i=1$. \\
\textbf{Case $i=2$: } For $\ell \in B_2$, define $\bZ_\ell = 1-\bY_\ell$. Now, applying the same argument as $i=1$, one obtains  
$$
\dtv\bigg(\sum_{\ell \in S_{opt,2}} \bZ_\ell, \sum_{\ell \in S_{2}} \bZ_\ell\bigg) \le \frac{\epsilon}{50}. 
$$
Furthermore, by guarantee of \textsf{Pseudo-knapsack} (in Step 4 of the algorithm), we have 
$|S_2| = |S_{opt,2}|$. 
Combining this with the above equation, we obtain  (\ref{eq:diff-3}) for $i=2$. \\
\textbf{Case $i=3$:} By the guarantees of the \textsf{Pseudo-knapsack} (in Step 5 of the algorithm), it follows that for every $j \le T_0$, 
$$
\sum_{\ell \in S_{3}} q_\ell^j = \sum_{\ell \in S_{opt,3}} q_\ell^j. 
$$
Using Lemma~\ref{lem:moment-matching}, it follows that
\[
\dtv\bigg(\sum_{\ell \in S_{opt,3}} \bY_\ell, \sum_{\ell \in S_{3}} \bY_\ell\bigg) \le 13 \cdot (T_0+1)^{\frac14} \cdot 2^{-(T_0+1)/2}. 
\]
Plugging in $T_0 = 4\log(1/\epsilon)$  and assuming $\epsilon$ is sufficiently small, we obtain (\ref{eq:diff-3}) for $i=3$. \\
\textbf{Case $i=4$:} For $\ell \in B_4$, define $\bZ_\ell = 1-\bY_\ell$. Applying the same argument as the case $i=3$, 
we obtain 
$$
\dtv\bigg(\sum_{\ell \in S_{opt,4}} \bZ_\ell, \sum_{\ell \in S_{4}} \bZ_\ell\bigg) \le \frac{\epsilon}{50}. 
$$
However, by guarantee of \textsf{Pseudo-knapsack} (in Step 6 of the algorithm), $|S_{opt,4}| = |S_4|$. Combining this, we obtain (\ref{eq:diff-3}) for $i=4$.
This finishes the proof of (\ref{eq:diff-3}). Next, we claim that for $1 \le i \le 4$, 
\begin{equation}~\label{eq:diff-4}
\dtv \bigg( \sum_{\ell \in S_{opt,i}} \bX_\ell, \sum_{\ell \in S_{opt,i}} \bY_\ell\bigg) \le \frac{\epsilon}{4} \  \textrm{and} \ \dtv \bigg( \sum_{\ell \in S_{i}} \bX_\ell, \sum_{\ell \in S_{i}} \bY_\ell\bigg) \le \frac{\epsilon}{4}. 
\end{equation}
We will only prove the first inequality, the proof of the second one is exactly the same. For $i=1,2$, (\ref{eq:diff-4}) follows from the fact that $\dtv(\bX_\ell, \bY_\ell) \le \frac{\epsilon}{4n}$ (for $i \in B_1, B_2$) and $|S_{opt,1}| + |S_{opt,2}| \le n$. For $i=3,4$, we claim that $|S_{opt,3}|, |S_{opt,4}| \le 100/\epsilon^3$. To see this, note that  $\sum_{\ell \in S_{opt,3}} q_\ell \le 1/\epsilon^2$ and on the other hand, for all $\ell \in B_3$, $q_\ell \ge \epsilon/100$. This implies $|S_{opt,3}| \le 100/\epsilon^3$. The proof for $|S_{opt,4}|$ is analogous. However, for $\ell \in B_3, B_4$, $\dtv(\bX_\ell, \bY_\ell) \le \frac{\epsilon^4}{1000}$. Thus, $\dtv(\sum_{\ell \in S_{opt,i}} \bX_\ell, \sum_{\ell \in S_{opt,i}} \bY_\ell) \le \frac{\epsilon}{10}$ for $i \in \{3,4\}$. This proves (\ref{eq:diff-4}). Combining (\ref{eq:diff-3}) and (\ref{eq:diff-4}), we obtain that 
\begin{eqnarray*}
\dtv\bigg(\sum_{\ell \in S} \bX_\ell,  \sum_{\ell \in S_{opt}} \bX_\ell \bigg) &\le& \sum_{i=1}^4 \dtv\bigg(\sum_{\ell \in S_i} \bX_\ell,  \sum_{\ell \in S_{opt,i}} \bX_\ell \bigg) \\
&\le& \sum_{i=1}^4 \dtv\bigg(\sum_{\ell \in S_i} \bY_\ell,  \sum_{\ell \in S_{opt,i}} \bY_\ell \bigg) + \sum_{i=1}^4 \dtv\bigg(\sum_{\ell \in S_i} \bX_\ell,  \sum_{\ell \in S_{i}} \bY_\ell \bigg) \\ &+& \sum_{i=1}^4 \dtv\bigg(\sum_{\ell \in S_{opt,i}} \bX_\ell,  \sum_{\ell \in S_{opt,i}} \bY_\ell \bigg). 
\end{eqnarray*}
Applying (\ref{eq:diff-3}) and (\ref{eq:diff-4}), we get all the three terms on the right hand side are bounded by $\epsilon$ and thus $$\dtv(\sum_{\ell \in S} \bX_\ell, \sum_{\ell \in S_{opt}} \bX_\ell) \le 3\epsilon.$$ This proves $\Pr[\sum_{j \in S} \bX_j>C] \le p + 3\epsilon$ which finishes the proof.

\end{proofof}

\section{Proof of Theorem~\ref{thm:Bernoulli2}} 
The proof of this theorem will be quite similar to the proof of Lemma~\ref{clm:SK-small}. We  start with the setup. For every $\ell \in [n]$, define $\bY_\ell$ to be an independent $A$-valued random variables obtained by rounding the probabilities in $\bX_\ell$ to the nearest multiple of $\frac{\epsilon}{4nk}$. Divide the interval $[0,1]$ into $s=\lceil 4ek^3 \rceil$ equal sized intervals; call them $I_1, \ldots, I_s$. We now define $\Phi\colon [n] \rightarrow [k] \times [s]^k$ as follows: (i) $\Phi_1(i) = \arg \max_{j} \Pr[\bX_i=a_j]$ (break  lexicographically if there is a tie).  (ii) For $1 \le j \le k$, $\Phi_{j+1}(i) = t$ if $\Pr[\bY_i= a_j] \in [(t-1)/s, t/s]$. The reason for defining the map $\Phi$ is simple: 
\begin{itemize}
\item For any $\beta \in [k] \times [s]^k$ and $1 \le j \le k$,  $|\max_{i \in \Phi^{-1}(\beta)} \Pr[\bX_i = a_j]- \min_{i \in \Phi^{-1}(\beta)} \Pr[\bX_i = a_j]| \le \frac1s \le \frac{1}{4ek^3}$. 
\item For any $\beta \in [k] \times [s]^k$ and any subset $S \subseteq \Phi^{-1}(\beta)$, $\sum_{\ell \in S} \Pr[\bX_i = \beta_1] \ge \frac{|S|}{k}$. 
\end{itemize}
Thus, this meets the first two conditions of Corollary~\ref{corr:PMD-moment-matching}. 
We now define $w \in \mathbb{Z}$ as $w = \lceil \log (16k s^k/\epsilon) \rceil$. Recall that $V_k(w)$ was defined as $V_k(w) = \{v \in \mathbb{Z}^k: v_i \ge0 \ \wedge \ \sum_i v_i \le w\}$. For every $\alpha \in V_k(w)$, let $A_\alpha$ denote the set defined as:
\[
A_\alpha= \bigg\{ j \cdot \frac{\epsilon^{\Vert \alpha \Vert_1}}{(4nk)^{\Vert \alpha \Vert_1}} : j \cdot \frac{\epsilon^{\Vert \alpha \Vert_1}}{(4nk)^{\Vert \alpha \Vert_1}} \le n \textrm{ and } j \in \mathbb{N}\bigg\}. 
\]
Note that the set $A_\alpha$ only depends on $\Vert \alpha \Vert_1$. We are naming $A_{\alpha}$ using $\alpha$ for notation reasons.

Finally, as in Claim~\ref{clm:SK-large} and Claim~\ref{clm:SK-small}, we will run the routine \textsf{Pseudo-knapsack}. For the routine, we define items 
 $\{J_\ell\}_{\ell \in [n]}$ as follows: Its ``size" is given by a $|V_k(w)|$-dimensional vector indexed by elements of $V_k(w)$. In particular, for $\alpha \in V_k(w)$, the $\alpha^{th}$ coordinate, denoted by $J_{\ell, \alpha}$ is given by $\prod_{j=1}^k \Pr[\bY_\ell=j]^{\alpha_j}$. Observe that crucially, for any subset $S \subseteq [n]$, $\sum_{\ell \in S} J_{\ell, \alpha} \in A_{\alpha}$. Further, we define the profit of $J_{\ell}$ to be $v_\ell$. 
 We are now ready to define the algorithm.  
\begin{enumerate}
\item Set $V_{\max} =0$ and $S_{\max}=\phi$. 
\item For $\{x_{\beta, \alpha} \in A_{\alpha}\}_{\beta \in [k] \times [s]^k, \alpha \in V_k(w) }$, 
\item  \hspace*{8pt} Run \textsf{Pseudo-knapsack} 
with items $\{J_{\ell}\}_{\ell \in \Phi^{-1}(\beta)}$ 
with target vector $x_{\beta}$ of dimension $|V_k(w)|$ 
and  \hspace*{11pt}the $\alpha^{th}$ coordinate is $x_{\beta, \alpha}$. 
\item \hspace*{8pt} The quantization list is a vector of dimension $|V_k(w)|$ whose $\alpha^{th}$ coordinate is $\frac{\epsilon^{\Vert \alpha \Vert_1}}{(4nk)^{\Vert \alpha \Vert_1}}$.  
\item \hspace*{8pt} Let the output of sets be $S_{\beta}$. 
Let $S = \cup_{\beta} S_{\beta}$. 
 \item \hspace*{8pt} If $\Pr[\sum_{\ell \in S} \bX_j >C] \le_{\epsilon/4} p + \frac{3\epsilon}{4}$ and $V_{\max} \leq \sum_{\ell \in S} v_\ell$, then set $S_{\max} \leftarrow S$ and $V_{\max} \leftarrow V$. 
\end{enumerate}
First, we bound the running time of this algorithm. Note that  the size of $A_{\alpha}$ is $ n \cdot \frac{(4nk)^{\Vert \alpha \Vert_1}}{\epsilon^{\Vert \alpha \Vert_1}}$. Further, $\Vert \alpha \Vert_1 \le w = O(\log (1/\epsilon) + k \log k)$. Thus, we have that for all $\alpha \in V_k(w)$, 
\[
|A_{\alpha}| = \bigg( \frac{nk}{\epsilon} \bigg)^{O(k \log k + \log(1/\epsilon))}. 
\]
As the total size of $|V_k(w)| \le w^k = O(  \log(1/\epsilon) +   k \cdot \log k)^k$, by Theorem~\ref{thm:multidimensional}, this bounds the running time to 
\[
 \bigg( \frac{nk}{\epsilon} \bigg)^{O(k \log k + \log(1/\epsilon)) \cdot (\log(1/\epsilon) + k \cdot \log k)^k} =\bigg( \frac{nk}{\epsilon} \bigg)^{O(k \log k + \log(1/\epsilon))^{k+1}}
\]
Recall that for $V \ge 0$, we define $\feas_{p,V}$ as 
\[
\feas_{p,V} = \big\{S \subseteq [n]: \Pr[\sum_{\ell \in S} \bX_\ell >C] \leq p \textrm{ and } \sum_{\ell \in S} v_\ell \ge V\big\}.
\]
Let $V_{opt} = \max \{V: \feas_{p,V} \textrm{ is not empty}\}$.  To prove the correctness of the algorithm, it suffices to show that there exists a choice of $\{x_{\beta,\alpha} \in A_{\alpha}\}_{\beta \in [k] \times [s]^k, \alpha \in V_k(w) }$ such that if the corresponding sets returned as $\{S_{\beta}\}_{\beta \in [k] \times [s]^k }$, then   for $S = \cup_{\beta} S_{\beta}$ (i) $\sum_{\ell \in S} v_\ell \ge V$ and (ii) $\Pr [ \sum_{\ell \in S} \bX_\ell >C ] \leq p + \epsilon/2$. 

Set $V = V_{opt}$ and let $S_{opt} \in \feas_{p,V}$. For $\beta \in [k] \times [s]^k$, define $S_{opt, \beta} = S_{opt} \cap \Phi^{-1}(\beta)$. For $\beta \in [k] \times [s]^k$ and $\alpha \in V_k(w)$, 
we define 
\[
x_{\beta, \alpha} = \sum_{\ell \in S_{opt, \beta}} \prod_{j=1}^k \Pr[\bY_{\ell} =a_j]^{\alpha_j}. 
\]
Also, define $V_{\beta} = \sum_{\ell \in S_{opt,\beta}}
v_\ell$.  Then, by guarantee of the routine \textsf{Pseudo-knapsack}, for this choice of $\{x_{\beta,\alpha}\}$, we obtain sets $\{S_{\beta}\}_{\beta \in [k] \times [s]^k}$ such that for all $\beta \in [k] \times [s]^k$ and $\alpha \in V_k(w)$ which satisfy  
\[
x_{\beta,\alpha}  = \sum_{\ell \in S_{\beta}} \prod_{j=1}^k \Pr[\bY_\ell=a_j]^{\alpha_j} \quad \textrm{and} \quad S_{\beta} \subseteq \Phi^{-1}(\beta). 
\]
Further, $\sum_{\ell \in S_{\beta}} v_\ell \ge \sum_{\ell \in S_{opt, \beta}} v_\ell$. Now, we apply Corollary~\ref{corr:PMD-moment-matching} on the partial sums $\sum_{\ell \in S_{\beta}} \bY_\ell$ and $\sum_{\ell \in S_{opt, \beta}} \bY_\ell$ to obtain that 
$$
\dtv(\sum_{\ell \in S_{\beta}} \bY_\ell, \sum_{\ell \in S_{opt,\beta}} \bY_\ell) = 2^{-w+1}  \le \frac{\epsilon}{16 k s^k}. 
$$
Adding this inequality over all $\beta \in [k] \times [s]^k$, we get
$$
\dtv(\sum_{\ell \in S} \bY_\ell, \sum_{\ell \in S_{opt}} \bY_\ell) \le \frac{\epsilon}{16}. 
$$
Further, by our rounding, for all $\ell \in [n]$, 
$\dtv(\bX_\ell, \bY_\ell) \le \frac{\epsilon}{4n}$. Thus, it immediately follows that $$\dtv(\sum_{\ell \in S} \bX_\ell, \sum_{\ell \in S_{opt}} \bX_\ell) \le \frac{\epsilon}{16} + \frac{\epsilon}{4}  < \frac{\epsilon}{2}.$$ This finishes the proof. 
\qed
\section{Proof of Theorem~\ref{thm:hyper-c} and Theorem~\ref{thm:hyper-main-bounded}}
We first start by sketching a reduction from Theorem~\ref{thm:hyper-c} to Theorem~\ref{thm:hyper-main-bounded}. As we have said before, this reduction is quite standard and follows the usual reduction which is used to obtain a polynomial time approximation scheme for the (deterministic) knapsack problem using the pseudopolynomial time algorithm. We give the reduction here for the sake of completeness. 

\subsubsection*{Reduction to the case when profits are small integers}
Given any class of random variables $\mathcal{D}$ supported on non-negative reals, the next lemma (Lemma~\ref{thm:reduction}) gives reduction from  an $(\epsilon, \epsilon)$ approximation  for stochastic knapsack instances of type $(\mathcal{D}, \mathbb{Q}^+)$ to an $(\epsilon,0)$ approximation for stochastic knapsack instance of type $(\mathcal{D}, \mathbb{Z}_M^+)$ where $M = \mathsf{poly}(n/\epsilon)$. In particular, this reduces Theorem~\ref{thm:hyper-c} to Theorem~\ref{thm:hyper-main-bounded}. 
\begin{lemma}\label{thm:reduction}
Let there be an algorithm $\mathcal{A}$ which given a stochastic knapsack instance of type $(\mathcal{D}, \mathbb{Z}_M)$ produces an $(\epsilon,0)$ approximation running in time $T(n,M, \epsilon)$. Then, there is an algorithm which given a stochastic knapsack instance of type $(\mathcal{D}, \mathbb{Q}^+)$ produces an $(\epsilon, \epsilon)$ approximation running in time 
$T(n, \mathsf{poly}(n/\epsilon), \epsilon)$. 
\end{lemma}
\begin{proof}
This proof follows the usual reduction from the approximation scheme for (standard) knapsack problem to 
the pseudopolynomial time algorithm for the knapsack problem. We sketch it here for the sake of completeness. Let the knapsack instance be given by items $\{I_j\}_{j=1}^n$ with profits $v_j$ and size $\bX_j$. Let the knapsack capacity be $C$ and the risk tolerance be $p$. Assume that the items are arranged so that $v_1 \le \ldots \le v_n$. Let $S^\ast$ be the optimal solution i.e., $S^\ast \in \feas_p$ and 
$ \mathsf{OPT} = \sum_{j \in S^\ast} v_j = \max_{S \in \feas_p} \sum_{j \in S} v_j$. 

Now, let $\ell_0$ be the largest index such that $\Pr[\bX_{\ell_0}>C] \le p$. Then, $v_{\ell_0} \le \mathsf{OPT} \le n \cdot v_{\ell_0}$. Clearly, for all $j>\ell_0$, $\bX_j \not \in S^\ast$ and thus, we can remove these items from our list. Let us define $K = \frac{\epsilon v_{\ell_0}}{n}$ and for all $1 \le \ell \le \ell_0$, define $w_\ell = \lfloor v_\ell /K \rfloor$. Note that $\{w_\ell\}_{\ell=1}^{\ell_0}$ are non-negative integers bounded by $M =\lceil n /\epsilon \rceil$. Let us define items $\tilde{I}_1, \ldots , \tilde{I}_{\ell_0}$  where $\tilde{I}_j = \{(\bX_j, w_j) \}$ and run $\mathcal{A}$ on this instance with overflow probability $p$ and capacity constraint $C$.
Also for $q>0$, let us define $\widetilde{\feas}_q$ as
$$
\widetilde{\feas}_q = \{S \subseteq [\ell_0]: \Pr[\sum_{j \in S} \bX_j >C] \le q\}. 
$$
 By guarantee of the algorithm $\mathcal{A}$, we output $S_w \in \widetilde{\feas}_{p+\epsilon}$ such that 
\begin{equation}~\label{eq:s-w-2}
\sum_{i \in S_w} w_i = \max_{\tilde{S} \in \widetilde{\feas}_{p}} \sum_{i \in \tilde{S}} w_i. 
\end{equation}
The final output is $S_w$. We now verify the guarantees of this algorithm. First, since $M =\mathsf{poly}(n/\epsilon)$, the running time of the algorithm is $T(n, \mathsf{poly}(n/\epsilon), \epsilon)$. Next, note that by definition, $\widetilde{\feas_{p+\epsilon}} \subseteq  \feas_{p+\epsilon}
$ for any $q>0$. As a consequence, 
\begin{equation}~\label{eq:feas-1}
S_w \in \feas_{p+\epsilon}. 
\end{equation}
Finally, to lower bound $\sum_{i \in S_w} v_i$, we make two observations. First is that $\widetilde{\feas_p} = \feas_p$. As a consequence, by definition, 
$$
\mathsf{OPT} = \max_{\tilde{S} \in \tilde{\feas_p}} \sum_{i \in \tilde{S}} v_i.
$$
Let us assume that $S_v$ achieves the optimum in the above equation. In other words, $\sum_{i \in S_v} v_i = \mathsf{OPT}$. Note that for every $i$, $v_i < K w_i + K$. 
Thus, we have 
$$
\mathsf{OPT}= \sum_{i \in S_v} v_i \le \sum_{i \in S_v} K \cdot w_i + K \cdot |S_v| \le  \sum_{i \in S_v} K \cdot w_i + \epsilon \cdot v_{\ell_0}. 
$$
The last inequality uses that $|S_v| \le n$. Now, using $\mathsf{OPT} \ge v_{\ell_0}$, we have \begin{equation}~\label{eq:s-opt-1} \sum_{i \in S_v} K \cdot w_i \ge (1-\epsilon) \cdot \mathsf{OPT}.
\end{equation}
Next, we observe that since $S_v \in \widetilde{\feas}_p$, using (\ref{eq:s-w-2}), 
$
\sum_{i \in S_w} w_i \ge \sum_{i \in S_v} w_i. 
$ As a result, using (\ref{eq:s-opt-1}), we get 
$\sum_{i \in S_w} K \cdot w_i \ge (1-\epsilon) \cdot \mathsf{OPT}$. However, note that for every $\ell \in [\ell_0]$, $v_\ell \ge w_\ell \cdot K$. Thus, we get that $\sum_{i \in S_w} v_i \ge (1-\epsilon) \cdot \mathsf{OPT}$. This finishes the proof. 

\end{proof}
%
We now turn to the proof Theorem~\ref{thm:hyper-main-bounded}. 
\subsection{Proof of Theorem~\ref{thm:hyper-main-bounded}} 
We start with some useful definitions. The important notion that we use here is the notion of critical index. This is an extension of the notion of critical index introduced by Servedio~\cite{Servedio:07cc} which has proved to be very influential in the complexity theoretic study of Boolean functions such as halfspaces and polynomial threshold functions. 
\begin{definition}~\label{def:critical-index}
Let $\bX_1, \ldots, \bX_n$ be a set of independent $(c,2,4)$-hypercontractive random variables and (are numbered so that) $\mu_2(\bX_1) \ge \ldots \ge \mu_2(\bX_n)$. For $\epsilon>0$, 
define the $\epsilon$-critical index of this sequence to be the smallest $1 \le i \le n$ such that 
$$
\frac{\mu_2(\bX_i)}{\sum_{j\ge i} \mu_2(\bX_j)} \le \frac{\epsilon^2}{c^4}.
$$
In case no such $i$ exists, then we say that the $\epsilon$-critical index of the sequence is $\infty$. 
\end{definition}
For the rest of this section, let us define the quantity $L(c,\epsilon)$ as $L(c,\epsilon) = (c^4/\epsilon^2) \log (1/\epsilon)$. 
\begin{definition}~\label{def:critical-index-1}
Let $\{(\bX_i,v_i)\}$ be a stochastic knapsack instance of type $(\mathcal{D}_c, \mathbb{Z}_M^+)$. For a subset $S \subseteq [n]$ and a parameter $\epsilon>0$, we define its $\epsilon$-type of $S$ as follows: Let $S = \{j_1, \ldots, j_{R}\}$ and let $K$ be the $\epsilon$-critical index of the set $\{\bX_{j_1}, \ldots, \bX_{j_R}\}$. Let $L = L(c,\epsilon)$. If $K < L$, the $\epsilon$-type is the tuple $(K,j_1, \ldots, j_K)$, else it is $(L, j_1, \ldots, j_{L})$.  
\end{definition}
To prove Theorem~\ref{thm:hyper-main-bounded}, it suffices to prove the following lemma.  
\begin{lemma}~\label{lem:etype}
There is an algorithm \textsf{SK-hyper} with the following guarantee: Let $\{(\bX_i, v_i)\}_{i=1}^n$ be a given stochastic knapsack instance of type $(\mathcal{D}_c, \mathbb{Z}_M^+)$, $ 1 \le V \le M \cdot n$, overflow probability $p$,  error parameter $\epsilon>0$, capacity $C$ and a given $\epsilon$-type 
$\mathcal{B} = (T, j_1, \ldots, j_T)$,  define the set $\feas_{p,\mathcal{B}, V}$ as
$$
\feas_{p,\mathcal{B}, V} = \big\{S \subseteq [n]: \sum_{i \in S} v_i = V, \epsilon\textrm{-type of }S \textrm{ is } \mathcal{B} \textrm{ and } \Pr[\sum_{i \in S} \bX_i >C] \le p\big\}. 
$$
If $\feas_{p,\mathcal{B}, V}$ is non-empty, the algorithm outputs $S^\ast$ such that $S^\ast \in \feas_{p+4\epsilon,\mathcal{B}, V}$ and runs in time 
$\mathsf{poly}(n, M, 1/\epsilon)$. 
\end{lemma}
We now see how Lemma~\ref{lem:etype} implies Theorem~\ref{thm:hyper-main-bounded}. 
\begin{proofof}{Theorem~\ref{thm:hyper-main-bounded}}
Let $\{(\bX_i,v_i)\}_{i=1}^n$ be the given knapsack instance of type $(\mathcal{D}_c, \mathbb{Z}_M^+)$ and let $\delta =\epsilon/8$. Let $\mathcal{C}_\delta$ be the set of all $\delta$-types for this instance. We now describe the algorithm: 
\begin{enumerate}
\item Initialize set $A$ to empty. 
\item For all $1 \le V \le n \cdot M$ and for all $\mathcal{B} \in \mathcal{C}_\delta$, 
\item \hspace{5pt} Run \textsf{SK-hyper} with $\delta$-type $\mathcal{B}$, error parameter $\delta$, overflow probability $p$ and profit $V$.
\item \hspace{5pt} Let the output set be $S^\ast_{\mathcal{B}, V}$. If $\Pr[\sum_{j \in S^\ast_{\mathcal{B}, V}} \bX_j >C ] \le_{2\delta} p +6 \cdot \delta$, add $S^\ast_{\mathcal{B}, V}$ to $A$. 
\item Output the set $S^\ast$ defined as 
$S^\ast =\arg \max_{S^\ast \in A} \sum_{j \in S^\ast} v_j$. 
\end{enumerate}

First, note that for any $S \in A$, $\Pr [ \sum_{j \in S} \bX_j>C] \le p + 8\delta$ i.e., $p +\epsilon$. Secondly, note if there exists any set $S$ such that $\Pr[\sum_{j \in S} \bX_j >C] \le p$ and $\sum_{j \in S} v_j = V$, then there exists $\mathcal{B} \in \mathcal{C}$, $\feas_{p, \mathcal{B}, V}$ is non-empty. 
In that case, we know  (by guarantee of \textsf{SK-hyper}) that there exists $S^{\ast}_{\mathcal{B},V} \in \feas_{p+4\delta, \mathcal{B},V}$ and $S^{\ast}_{\mathcal{B},V} \in A$. Thus, if $S^{\ast}$ is the output of the last step of the algorithm, $\sum_{j \in S^{\ast}} v_j \ge V$. Thus the output of the algorithm is an $(\epsilon,0)$ approximation. To bound the running time, note that every element of $\mathcal{C}$ is of the form $(T, j_1, \ldots, j_T)$ where $T \le L(c,\delta)$ and $\{j_1, \ldots, j_T\} \subseteq [n]$. Thus, $|\mathcal{C}_\delta| = n^{O(L(c,\delta))}$. For each $\mathcal{B} \in \mathcal{C}_\delta$, the running time of Steps 3, 4, 5 is bounded by $\mathsf{poly}(n,M)$. This shows that the total running time is $n^{\tilde{O}(c^4/\epsilon^2)} \cdot M^{O(1)}$, finishing the proof.


\end{proofof}

First, let us set $L = L(c,\epsilon)$. We will now split the proof of Lemma~\ref{lem:etype} 
into two parts: (a) One where the $\epsilon$-type is of the form $(L, j_1, \ldots, j_L)$ where $\{j_1, \ldots, j_L \} \subseteq [n]$. (b) In the second case, the $\epsilon$-type is of the form $(T, j_1, \ldots, j_T)$
where $T<L$. Lemma~\ref{lem:etype-large} covers case (a) and Lemma~\ref{lem:etype-small} covers case (b). 

\begin{lemma}~\label{lem:etype-large}
There is an algorithm \textsf{SK-hyper-large} with the following guarantee: Let $\{(\bX_i,v_i)\}_{i=1}^n$ be a stochastic knapsack instance of type $(\mathcal{D}_c, \mathbb{Z}_M^+)$ with capacity $C$, overflow probability $p$, error parameter $\epsilon>0$, profit $V$   and a given $\epsilon$-type $\mathcal{B} = (T, j_1, \ldots, j_T)$ where 
$T = L$. If the set $\feas_{p,\mathcal{B},V}$ is non-empty, then \textsf{SK-hyper-large} outputs $S^\ast_{\mathcal{B},V}$ such that $\sum_{j \in S^\ast_{\mathcal{B},V}} v_j = V$ and  $\Pr[\sum_{j \in S^\ast_{\mathcal{B},V}} \bX_j >C] < p + \epsilon$. The running time of the algorithm is $\mathsf{poly}(n,M,1/\epsilon)$. 
\end{lemma}

\begin{lemma}~\label{lem:etype-small}
There is an algorithm \textsf{SK-hyper-small} with the following guarantee: Let $\{(\bX_i,v_i)\}_{i=1}^n$ be a stochastic knapsack instance of type $(\mathcal{D}_c, \mathbb{Z}_M^+)$ with capacity $C$, overflow probability $p$, error parameter $\epsilon>0$, profit $V$   and a given $\epsilon$-type $\mathcal{B} = (T, j_1, \ldots, j_T)$ where 
$T < L$. If the set $\feas_{p,\mathcal{B},V}$ is non-empty, then \textsf{SK-hyper-large} outputs $S^\ast_{\mathcal{B},V}$ such that $\sum_{j \in S^\ast_{\mathcal{B},V}} v_j = V$ and  $\Pr[\sum_{j \in S^\ast_{\mathcal{B},V}} \bX_j >C] < p + \epsilon$. The running time of the algorithm is $\mathsf{poly}(n,M,1/\epsilon)$. 
\end{lemma}

\subsection{Proof of Lemma~\ref{lem:etype-large}}
Recall that $\mathcal{B} = \{L, j_1, \ldots, j_L\}$. 
We first define the set $\Gamma = \{i \in [n]: \mu_2(\bX_i) \le \mu_2(\bX_{j_{L-1}}) \}$. Let us define $S' = S \setminus \{j_1, \ldots, j_{L-1}\}$. 
Note that if the $\epsilon$-type of $S$ is $\mathcal{B}$, 
then $S' \subseteq \Gamma$.  Let us now define the rational number $\rho$ to be such that 
$$
\frac{n^2 \cdot \rho}{\epsilon^2} \le \mu_2(\bX_{j_{L-1}}) \le \frac{2\cdot n^2 \cdot \rho}{\epsilon^2}.
$$
Note that we can efficiently compute such a $\rho$ 
and is an integral multiple of $(\epsilon^2/n^2)  \cdot 2^{-3n}$. For $\ell \in \Gamma$, we define $\beta_\ell$ as $$
\beta_\ell = \bigg\lfloor \frac{\mu_2(\bX_\ell)}{\rho} \bigg\rfloor \rho. 
$$ 
In other words, $\beta_\ell$ is the integral multiple of $\rho$ which is closest to $\mu_2(\bX_\ell)$ (and larger than $\mu_2(\bX_\ell)$).  For every $\ell \in \Gamma$, we define item $J_\ell$ with ``size" $(v_\ell, \beta_\ell)$ and ``profit" $-\mathbf{E}[\bX_\ell]$. 
Also, let us define the set $A = \{0, \rho, \ldots, \frac{2c^4\cdot n^2}{\epsilon^4} \cdot \rho \}$. 
We now describe the algorithm. 
\begin{enumerate}
\item Let $\tilde{V} = V- (v_{j_1} + \ldots + v_{j_{L-1}})$. 
\item For all $x \in A$, 
\item \hspace*{5pt} Run \textsf{Pseudo-knapsack} on items $\{J_\ell\}_{\ell \in \Gamma}$ with target $(x,V)$ and quantization $(\rho,1)$. Let the output be $\tilde{S} \subseteq \Gamma$. 
\item \hspace*{5pt} Let $S = \tilde{S} \cup \{j_1, \ldots, j_{L-1}\}$. If $\Pr[\sum_{j \in S} \bX_j >C] \le_{\epsilon/4} p + 3\epsilon/4$, output $S$. 
\end{enumerate}
First of all, observe that if the algorithm outputs a set $S$ then it clearly satisfies the requirement. Further, for every $\ell \in \Gamma$, $\beta_\ell/\rho$ is a non-negative integer bounded by $2n^2/\epsilon^2$.
Applying the guarantee of Theorem~\ref{thm:multidimensional}, the running time is bounded by $\mathsf{poly}(n,M,1/\epsilon)$. Thus to prove the correctness of the algorithm, it suffices to prove that if $\feas_{p, \mathcal{B},V}$ is non-empty, then there exists $x \in A$ such that the output $S$ (corresponding to $x$) satisfies $\Pr[\sum_{j \in S} \bX_j>C] \le p +\epsilon/2$. To prove this, let us assume that $S_{opt} \in \feas_{p, \mathcal{B},V}$ and $\tilde{S}_{opt} = S \setminus \{j_1, \ldots, j_{L-1}\}$. Let us now define $x = \sum_{\ell \in \tilde{S}_{opt}} \beta_\ell $  and $y= -\sum_{\ell \in \tilde{S}_{opt}} \mathbf{E}[\bX_\ell] $. Note that $x$ is an integral multiple of $\rho$. Further, 
$$
\sum_{\ell \in \tilde{S}_{opt}} \beta_\ell \le \sum_{\ell \in \tilde{S}_{opt}} \mu_2(\bX_\ell) \le \frac{c^4}{\epsilon^2} \cdot \mu_2(\bX_{L-1})  \le \frac{c^4}{\epsilon^2} \cdot \frac{2n^2\cdot \rho}{\epsilon^2} = \frac{2n^2  \cdot c^4 \cdot \rho }{\epsilon^4}. 
$$
Thus, $x$ lies in the set $A$. 
 By guarantee of the routine \textsf{Pseudo-knapsack}, the output is a set $\tilde{S}$ with the following properties: 
\begin{enumerate}
\item[(i)] $\tilde{S} \subseteq \Gamma$, $\sum_{\ell \in \tilde{S}} v_\ell = {V}$ and $\sum_{\ell \in \tilde{S}} \beta_\ell = x $. 
\item[(ii)]~\label{ea:ab} $-\sum_{\ell \in \tilde{S}} \mathbf{E}[\bX_\ell]  \ge -\sum_{\ell \in \tilde{S}_{opt}} \mathbf{E}[\bX_\ell]$. 
\end{enumerate}
Let us now observe that for all $\ell \in \Gamma$, $
\mu_2(X_\ell) \le (\beta_\ell + \rho)$ and thus, 
\begin{equation}~\label{eq:variance-bound-2}
\sum_{\ell \in \tilde{S}} \mu_2(X_\ell) \le \sum_{\ell \in \tilde{S}} (\beta_\ell + \rho)  \le \frac{\epsilon^2 \cdot \mu_2(\bX_{j_{L-1}})}{n} + \sum_{\ell \in \tilde{S}} \beta_\ell \le \frac{\epsilon^2 \cdot \mu_2(\bX_{j_{L-1}})}{n} + \sum_{\ell \in \tilde{S}_{opt}} \mu_2(X_\ell). 
\end{equation}
\begin{claim}~\label{clm:critical}
The $2\epsilon$-type of the set $S$ is $\mathcal{B'}= (L, i_1, \ldots, i_{L} )$ where for $1 \le k \le L-1$, 
$i_k = j_k$. 
\end{claim}
\begin{proof}
Let us enumerate $S = \{j'_1, \ldots, j'_{R}\}$ such that $\mu_2(\bX_{j'_1}) \ge \ldots \ge \mu_2(\bX_{j'_{R}})$.
Observe that by construction, for every $\ell \in \tilde{S}$, $\mu_2(\bX_\ell) \le \mu_2(\bX_{L-1})$.  Thus, 
$j'_1 = j_1$, $j'_2 = j_2  \ \ldots \ j'_{L-1} = j_{L-1}$. 
Next, observe that for any $1 \le i \le L-1$, 
\begin{eqnarray*}
\frac{\mu_2(\bX_{j'_i})}{\sum_{k \ge i} \mu_2(\bX_{j'_{k}})} &=&\frac{\mu_2(\bX_{j_i})}{\sum_{L > k \ge i} \mu_2(\bX_{j_{k}}) + \sum_{k \ge L} \mu_2(\bX_{j'_{k}})}  = \frac{\mu_2(\bX_{j_i})}{\sum_{L > k \ge i} \mu_2(\bX_{j_{k}}) + \sum_{\ell \in \tilde{S}} \mu_2(\bX_{\ell})} \\  &\geq&  \frac{\mu_2(\bX_{j_i})}{\sum_{L > k \ge i} \mu_2(\bX_{j_{k}}) + \sum_{\ell \in \tilde{S}_{opt}} \mu_2(\bX_{\ell}) + \frac{\epsilon^2 \cdot \mu_2(\bX_{j_{L-1}})}{n}} \\ 
  &\geq&  \frac{\mu_2(\bX_{j_i})}{\sum_{L > k \ge i} \mu_2(\bX_{j_{k}}) + \sum_{\ell \in \tilde{S}_{opt}} \mu_2(\bX_{\ell}) + \frac{\epsilon^2 \cdot \mu_2(\bX_{j})}{n}}
\end{eqnarray*}
Here the first inequality uses (\ref{eq:variance-bound-2}). Now, since the $\epsilon$-type of $S$ is $\{L, j_1, \ldots, j_L\}$, we have that 
$$
\frac{\mu_2(\bX_{j_i})}{\sum_{L > k \ge i} \mu_2(\bX_{j_{k}}) + \sum_{\ell \in \tilde{S}_{opt}} \mu_2(\bX_{\ell})}  \ge \frac{\epsilon^2}{c^4}. 
$$
This immediately implies the claim. 
\end{proof}
We now state the following important proposition. 
\begin{proposition}~\label{prop:anti}
Let $\bX_1, \ldots, \bX_n$ be a sequence of independent $(c,2,4)$ hypercontractive random variables. Further, the $\epsilon$-critical index of this sequence is at least $L=L(c,\epsilon)$ (as in Definition~\ref{def:critical-index-1}). Define another sequence of random variables such $\bY_1, \ldots, \bY_n$ such that 
$\bY_i = \bX_i$ (if $i<L$) and $\bY_i = \mathbf{E}[\bX_i]$ (otherwise). Then, 
$$
\dcdf(\sum_{i=1}^n \bX_i , \sum_{i=1}^n \bY_i) \le \frac{\epsilon}{8}. 
$$
\end{proposition}
Before we see the proof of this proposition, let us first see why this implies the correctness of our algorithm. To see this, for every $\ell \in \Gamma$, let us define the random variable $\bY_\ell$ to be $\mathbf{E}[\bX_\ell]$ with probability $1$. For $\ell \not \in \Gamma$, $\bX_\ell = \bY_\ell$. 
Then, applying Proposition~\ref{prop:anti} to the sequence 
$\{\bX_{j_1}, \ldots, \bX_{j_T}\}$, we get
\begin{equation}~\label{eq:hybrid-1}
\dcdf( \sum_{\ell \in S_{opt}} \bX_{\ell}, \sum_{\ell \in S_{opt}} \bY_\ell) \le \frac{\epsilon}{8}. 
\end{equation}
Likewise, if we enumerate $S = \{j'_1, \ldots, j'_R\}$ and apply Proposition~\ref{prop:anti} to the sequence 
$\{\bX_{j'_1}, \ldots, \bX_{j'_R}\}$, we get 
\begin{equation}~\label{eq:hybrid-2}
\dcdf( \sum_{\ell \in S} \bX_{\ell}, \sum_{\ell \in S} \bY_\ell) \le \frac{\epsilon}{4}. 
\end{equation}
Finally, note that \begin{equation}~\label{eq:hybrid-3} \sum_{\ell \in S} \bY_\ell - \sum_{\ell \in S_{opt}} \bY_\ell = \sum_{\ell \in \tilde{S}} \bY_\ell - \sum_{\ell \in \tilde{S}_{opt}} \bY_\ell = \sum_{\ell \in \tilde{S}} \mathbf{E}[\bX_\ell] - \sum_{\ell \in \tilde{S}_{opt}} \mathbf{E}[\bX_\ell]<0. \end{equation} 
The second equality uses the definition of the random variables $\{\bY_\ell\}$ and the last inequality uses the guarantee of \textsf{Pseudo-knapsack}. 
Now, applying (\ref{eq:hybrid-2}), we get that for every $t \in \mathbb{R}$, 
$
\Pr[\sum_{\ell \in S} \bX_\ell >C] \le \Pr[\sum_{\ell \in S} \bY_\ell >C] + \frac{\epsilon}{4}. 
$ However, applying (\ref{eq:hybrid-3}), this implies 
$
\Pr[\sum_{\ell \in S} \bX_\ell >C] \le \Pr[\sum_{\ell \in S_{opt}} \bY_\ell >C] + \frac{\epsilon}{4}. 
$ Finally, applying (\ref{eq:hybrid-1}), this finally implies that 
$\Pr[\sum_{\ell \in S} \bX_\ell >C] \le \Pr [\sum_{\ell \in S_{opt}} \bX_\ell >C] + \frac{3\epsilon}{8}$.  As $\Pr [\sum_{\ell \in S_{opt}} \bX_\ell >C] = p$, we obtain
$$
\Pr[\sum_{\ell \in S} \bX_\ell >C] \le p + \frac{3\epsilon}{8}. 
$$
This concludes the proof modulo Proposition~\ref{prop:anti} which we prove next. 
\begin{proofof}{Proposition~\ref{prop:anti}}
Let us set $K = \frac{2^{25} \cdot (c+2)^4}{9 \cdot \epsilon^2}$ (The choice of $2^{25}$ is not crucial and can be essentially any large constant). Define the random variable $\bZ_K   = \sum_{i \le K} \bX_i$ 
and let $\sigma_K  = \sqrt{\mu_2(\bX_K)}$. Note that $\{\bX_i\}$ is arranged in decreasing order of variance. Thus, applying Lemma~\ref{lem:hypercontractive}, we get that for all $1 \le j \le K$, 
\begin{equation}~\label{eq:anti-conc}
Q_{\bX_j}\bigg(\frac{\sigma_K}{2}\bigg)  \le Q_{\bX_j}\bigg(\frac{\sqrt{\mu_2(\bX_j)}}{2}\bigg) 
\le 
1 - \frac{9}{128(c+2)^4}. 
\end{equation}
Using the above and Lemma~\ref{lem:Kolmogorov} (on the variable $\bZ_K$), we get 
\begin{equation}~\label{eq:cheb-2}
Q_{\bZ_K}(\sigma_K/2) \le \frac{100 \sigma_K}{\sqrt{\sum_{i=1}^K \sigma_K^2 \cdot \frac{9}{128(c+2)^4}}} =  \frac{100 }{\sqrt{K \cdot \frac{9}{128(c+2)^4}}} \le \frac{\epsilon}{16}.
\end{equation}



Define $\bZ_L = \sum_{j \le L} \bX_j$. Since $L=L(c,\epsilon) \ge K$, (applying Fact~\ref{fact:convolve}) we  obtain that
$$
Q_{\bZ_L}(\sigma_K/2) \leq \frac{\epsilon}{16}. 
$$
Next, define the random variable $\bZ_{>L}$ as $\bZ_{>L} = \sum_{j>L} \bX_j$. Using $L -K \ge (2c^4/\epsilon^2)\cdot \log(1/\epsilon)$ and the definition of critical index, we have 
\[
\mu_2(\bZ_{>L}) \le (1-\epsilon)^{L-K} \cdot \sigma_K^2 \le \epsilon^2 \cdot \sigma_K^2. 
\]
By applying Chebyshev's inequality, 
\begin{equation}~\label{eq:cheb-1}
\Pr\bigg[|\bZ_{>L} -  \mathbf{E}[\bZ_{>L}| \ge \frac{\sigma_K}{8}\bigg] \le O(\epsilon^2). 
\end{equation}
%
Now note that $\sum_{i=1}^n \bX_i = \bZ_{>L} + \bZ_{L}$ 
and $\sum_{i=1}^n \bY_i = \mathbf{E}[\bZ_{>L}] + \bZ_{L}$. 
Now, consider any $x \in \mathbb{R}$. Then, 
$\mathsf{sign} (\sum_{i=1}^n \bX_i - x) \not = \mathsf{sign} (\sum_{i=1}^n \bY_i - x)$  only if at least one of the following happens: (i) $|\bZ_{>L} -  \mathbf{E}[\bZ_{>L}]| \ge \frac{\sigma_K}{8}$  (ii) $|\bZ_{L} - x| \le \frac{\sigma_K}{4}$. Using (\ref{eq:cheb-1}) and (\ref{eq:cheb-2}), 
this implies  
\begin{eqnarray*}
\bigg| \Pr[\sum_{i=1}^n \bX_i \le x] - \Pr[\sum_{i=1}^n \bY_i \le x]\bigg| &\le&  \Pr\big[|\bZ_{>L} -  \mathbf{E}[\bZ_{>L}| \ge \frac{\sigma_K}{8}\big] + Q_{\bZ_L} (\sigma_K/2)  \\
&\leq& \frac{\epsilon}{16} + O(\epsilon^2) \leq  \frac{\epsilon}{8}. 
\end{eqnarray*}
\end{proofof}

\subsection{Proof of Lemma~\ref{lem:etype-small}}
The initial setup of this proof will be quite similar to the proof of Lemma~\ref{lem:etype-large}. However, since there are some subtle differences, we repeat the setup again. 
Recall that $\mathcal{B} = \{T, j_1, \ldots, j_T\}$  where $T<L$. We now define the set $\Gamma = \{ i \in [n]: \mu_2(\bX_i) \le \mu_2(\bX_T)\}$. Let $\rho$ be a rational number such that   
\begin{equation}~\label{eq:rational}
\frac{n^4 \cdot \rho}{\epsilon^4} \le \mu_2(\bX_{j_T}) \le \frac{2n^4 \cdot \rho}{\epsilon^4}. 
\end{equation} 
Note, we can efficiently compute such a $\rho$ which is an integral multiple of $(\frac{\epsilon^4}{n^4}) \cdot 2^{-3n}$. For $\ell \in \Gamma$, we define 
$\beta_\ell$ as
$$
\beta_\ell = \bigg\lfloor \frac{\mu_2(\bX_\ell)}{\rho}\bigg\rfloor \cdot \rho. 
$$
In other words, we obtain $\beta_\ell$ by rounding (down) $\mu_2(\bX_\ell)$ to the nearest multiple of $\rho$. Now, for every $\ell \in \Gamma$, we define item $J_\ell$ with ``size" $(v_\ell, \beta_\ell)$ and ``profit" $-\mathbf{E}[\bX_\ell]$. Finally, let us define the set 
$A = \{\rho \cdot j : j \in \mathbb{N} \textrm{ and } 
(c^4 \cdot n^4/\epsilon^6) - n \le j \le (2n^5/\epsilon^4)\}$. With this notation, we describe the algorithm (which is the same as the algorithm in Lemma~\ref{lem:etype-large}). 
\begin{enumerate}
\item Let $\tilde{V} = V- (v_{j_1} + \ldots + v_{j_{L-1}})$. 
\item For all $x \in A$, 
\item \hspace*{5pt} Run \textsf{Pseudo-knapsack} on items $\{J_\ell\}_{\ell \in \Gamma}$ with target $(x,V)$ and quantization $(\rho,1)$. Let the output be $\tilde{S} \subseteq \Gamma$. 
\item \hspace*{5pt} Let $S = \tilde{S} \cup \{j_1, \ldots, j_{L-1}\}$. If $\Pr[\sum_{j \in S} \bX_j >C] \le_{\epsilon/4} p + 3\epsilon/4$, output $S$.
\end{enumerate} 
As before, it is easy to see that that the running time of this procedure is bounded by $\mathsf{poly}(n,M, 1/\epsilon)$. Further, as in Claim~\ref{clm:SK-large}, if the algorithm outputs a set $S$, then it satisfies our requirement. Thus, to prove correctness of the algorithm, it suffices to show that if $\feas_{p, \mathcal{B}, V}$ is non-empty, then
there exists $x \in A$ such that the output $S$ (corresponding to $x$) satisfies $\Pr [\sum_{j \in S} \bX_j>C] \le p + \frac{\epsilon}{2}$. Let us assume that $S_{opt} \in \feas_{p,\mathcal{B},V}$ and define $\tilde{S}_{opt} = S \setminus \{j_1, \ldots, j_{T-1}\}$. Let us now define $x = \sum_{\ell \in \tilde{S}_{opt}} \beta_\ell$ and $y = \sum_{\ell \in \tilde{S}_{opt}} -\mathbf{E}[\bX_\ell]$. Now, as $\mu_2(\bX_T) \ge \mu_2(\bX_{\ell})$ for any $\ell \in \tilde{S}_{opt}$, hence for all  $\ell \in \tilde{S}_{opt}$, $\beta_\ell \le \beta_{T}$. Thus, $x \le |\tilde{S}_{opt}| \cdot \beta_T \le \frac{2n^5}{\epsilon^4}$. 
On the other hand, we have 
\begin{eqnarray}~\label{eq:var-lower-bound}
\sum_{\ell \in \tilde{S}_{opt}} \beta_\ell \ge \sum_{\ell \in \tilde{S}_{opt}} \mu_2(\bX_\ell) - n \cdot \rho \ge \frac{c^4 \cdot \mu_2(\bX_T)}{\epsilon^2}  - n \cdot \rho \ge \frac{c^4 \cdot n^4 \cdot \rho}{\epsilon^6}  - n \cdot \rho. 
\end{eqnarray}
The second inequality uses that $T \in \tilde{S}_{opt}$ and the definition of $\epsilon$-type. 
As $x$ is trivially an integral multiple of $\rho$, combining with the above inequalities, we get that $x \in A$. By guarantee of the  routine \textsf{Pseudo-knapsack}, we get that there is a output set $\tilde{S}$ with the following properties: 
\begin{itemize}
\item[(i)] $\tilde{S} \subseteq \Gamma$, $\sum_{\ell \in \tilde{S}} v_\ell = V$ and $\sum_{\ell \in \tilde{S}} \beta_{\ell} =x$. 
\item[(ii)] $-\sum_{\ell \in \tilde{S}} \mathbf{E}[\bX_\ell] \ge -\sum_{\ell \in \tilde{S}_{opt}} \mathbf{E}[\bX_\ell]$. 
\end{itemize}
Next, we have that
\begin{eqnarray}\label{eq:tilde-2}
 \frac{\max_{\ell^\ast \in \tilde S}\mu_2(\bX_{\ell^\ast})}{\sum_{\ell \in \tilde S} \mu_2(\bX_\ell)} \le \frac{\max_{\ell^\ast \in \Gamma}\mu_2(\bX_{\ell^\ast})}{\sum_{\ell \in \tilde S} \beta_\ell} \leq \frac{\frac{2n^4 \cdot \rho}{\epsilon^4}}{\rho \cdot \frac{n^4 \cdot c^4}{\epsilon^6} - \rho \cdot n} \le \frac{3 \epsilon^2}{c^4}. 
\end{eqnarray}
The first inequality uses that for every $\ell$, $\mu_2(\bX_\ell) \ge \beta_\ell$ and $S \subseteq \Gamma$. The second inequality follows by applying (\ref{eq:var-lower-bound}) and (\ref{eq:rational}) along with the definition of $\Gamma$. Similarly, it also follows that
\begin{equation}~\label{eq:tilde-1}
 \frac{\max_{\ell^\ast \in \tilde{S}_{opt}}\mu_2(\bX_{\ell^\ast})}{\sum_{\ell \in \tilde{S}_{opt}} \mu_2(\bX_\ell)} \le \frac{3 \epsilon^2}{c^4}. 
\end{equation}
Let $\tilde{\sigma}$, $\tilde{\sigma}_{opt}$, $\tilde{\mu}$ and $\tilde{\mu}_{opt}$ be defined as
$$
\tilde{\sigma}^2  = \mu_2(\sum_{\ell \in \tilde{S}} \bX_\ell), \ \tilde{\mu} = \mathbf{E}[\sum_{\ell \in \tilde S} \bX_\ell], \ \tilde{\sigma}_{opt}^2  = \mu_2(\sum_{\ell \in \tilde{S}_{opt}} \bX_\ell), \ \tilde{\mu}_{opt} = \mathbf{E}[\sum_{\ell \in \tilde S_{opt}} \bX_\ell]. 
$$
Applying Corollary~\ref{corr:BE}  with (\ref{eq:tilde-2}) and (\ref{eq:tilde-1}), we obtain 
\begin{equation}~\label{eq:define-3}
\dcdf \big( \sum_{\ell \in \tilde{S}} \bX_\ell, \mathcal{N}(\tilde{\mu}, \tilde{\sigma}^2)\big) \le \sqrt{3} \cdot \epsilon, \ \ \dcdf \big( \sum_{\ell \in \tilde{S}_{opt}} \bX_\ell, \mathcal{N}(\tilde{\mu}_{opt}, \tilde{\sigma}_{opt}^2)\big) \le \sqrt{3} \cdot \epsilon
\end{equation}
Further, $\tilde{\sigma}^2$ and $\tilde{\sigma}_{opt}^2$ are close in the following sense:  
\[
\frac{|\tilde{\sigma}^2 - \tilde{\sigma}_{opt}^2|}{\tilde{\sigma}^2} \le\frac{ |\sum_{\ell \in \tilde S} \beta_\ell - \sum_{\ell \in \tilde S_{opt}} \beta_\ell +  \rho \cdot (|\tilde{S}| + \tilde{S}_{opt}|)|}{\sum_{\ell \in \tilde{S}} \beta_\ell} \le \frac{2\rho n}{\rho \cdot \big(\frac{c^4 n^4}{\epsilon^6} -n  \big)} \le  \frac{2\epsilon^6}{c^4 n^3}. 
\]
Thus, 
$
\dcdf(\mathcal{N}(\tilde{\mu}, \tilde{\sigma}^2), \mathcal{N}(\tilde{\mu}, \tilde{\sigma}_{opt}^2)) \le \frac{\sqrt{2} \cdot \epsilon^3}{c^2 n^{1.5}}
$. Finally, note that by guarantee of \textsf{Pseudo-knapsack}, we have $\tilde{\mu} \le \tilde{\mu}_{opt}$ and thus for all $t \in \mathbb{R}$, 
$
\Pr[\mathcal{N}(\tilde{\mu}_{opt}, \tilde{\sigma}_{opt}^2) \le t] \ge \Pr[\mathcal{N}(\tilde{\mu}, \tilde{\sigma}_{opt}^2) \le t] 
$. Combining this with (\ref{eq:define-3}), we obtain that for all $t \in \mathbb{R}$, 
\[
\Pr\big[\sum_{\ell \in \tilde{S}} \bX_\ell \le t\big] \le \Pr\big[\sum_{\ell \in \tilde{S}_{opt}} \bX_\ell \le t\big] + 2\sqrt{3} \epsilon + \frac{\sqrt{2} \cdot \epsilon^3}{c^2 n^{1.5}}
\]
 Adding the random variable $\bX_{j_1} + \ldots + \bX_{j_{L-1}}$ to both sides, we get 
 \[
\Pr\big[\sum_{\ell \in {S}} \bX_\ell \le t\big] \le \Pr\big[\sum_{\ell \in {S}_{opt}} \bX_\ell \le t\big] + 2\sqrt{3} \epsilon + \frac{\sqrt{2} \cdot \epsilon^3}{c^2 n^{1.5}} \le \Pr\big[\sum_{\ell \in {S}_{opt}} \bX_\ell \le t\big] + 4 \epsilon. 
\]
This proves the lemma. 

\section*{Acknowledgments} 
The author is extremely grateful to Costis Daskalakis, Gautam Kamath, Rocco Servedio and Christos Tzamos for their contributions to this project.  
The author also thanks Chandra Chekuri,  Vineet Goyal, Sanjeev Khanna and Jian Li for helpful email exchanges about the stochastic knapsack problem and David Morton for a very useful conversation about stochastic optimization.  Finally, the author thanks Rocco Servedio and Aravindan Vijayaraghavan for their many helpful comments concerning the presentation. 
\bibliography{allrefs}
\bibliographystyle{plain}

\appendix

\section{Pseudopolynomial time algorithm for multidimensional knapsack}
The well-known pseudopolynomial time algorithm for the (standard) multidimensional knapsack (see \cite{Kellerer2004} for a reference) will be one of our principal algorithmic tools. We recall the guarantee of this algorithm below. 
\begin{theorem}~\label{thm:multidimensional}
Let $\{J_\ell\}_{\ell=1}^n$ be a collection of items such that the ``size of" $J_\ell$ is $\overline{x}_\ell = (x_\ell^{(1)}, \ldots, x_\ell^{(k)}) \in \mathbb{R}^{+k}$ and profit $v_\ell \in \mathbb{R}$. Further, the item sizes are \emph{quantized} i.e., there are $\alpha= (\alpha_1, \ldots, \alpha_k)$ such that for all $1 \le \ell \le n$ and 
$1 \le j \le k$,  $x_\ell^{(j)}$ is an integral multiple of $\alpha_j$. For $\overline{y} = (y_1, \ldots, y_k)$, 
define the set $\mathcal{A}_{\overline{y}}$ 
as $\mathcal{A}_{\overline{y}} = \{S \subseteq [n]: \sum_{i \in S} (x_i^{(1)}, \ldots , x_i^{(k)}) = (y_1, \ldots, y_k)\}$. 
There is an algorithm \textsf{Pseudo-knapsack} 
such that given a \emph{target size} $(y_1, \ldots, y_k) \in \mathbb{R}^k$, the algorithm outputs $S^\ast \subseteq [n]$ such that $S^\ast \in \mathcal{A}_{\overline{y}}$ and $$
\sum_{ i \in S^\ast} v_i =\max_{S \in \mathcal{A}_{\overline{y}}} \sum_{i \in S} v_i.
$$
Assuming that $(y_j/\alpha_j) = M_j$, the running time of the algorithm is $\mathsf{poly}(n, \prod_{j \in k} M_j)$. 
In case, no profits are specified, the algorithm simply outputs a set $S \in \mathcal{A}$ if $\mathcal{A}$ is non-empty.  
\end{theorem}
\begin{sketch}
The proof is quite standard and follows the usual dynamic programming formulation used to obtain a pseudopolynomial time algorithm for the standard knapsack problem. We leave the details to the interested reader. 
\end{sketch}

\section{Hypercontractivity of well-known random variables}~\label{app:hyper}
Table~\ref{tab:hyper} lists some common $(c,2,4)$ hypercontractive random variables along with the explicit values for $c^4$.  We note that while for many random variables (such as Gaussian or Laplace), the value of $c$ is an absolute constant independent of the parameters,  in other cases, the value of $c$ depends on the parameters of the distribution (such as in the case of a Poisson). This directly affects the running time of Theorem~\ref{thm:hyper-c} where the exponent of $n$ is $\tilde{O}(c^4/\epsilon^2)$ if all the individual variables are $(c,2,4)$ hypercontractive.  
\begin{table}[h]
\centering
\label{tab:hyper}
\begin{tabular}{|c|c|}
\hline
Type of random variable &  value of $c^4$  \\
\hline
Gaussian & ${3}$  \\
 \hline
 Poisson $(\lambda)$ &  $3 + \frac{1}{\lambda}$ \\
 \hline
 Exponential &   $9$ \\
 \hline 
 Laplace & $6$ \\ 
 \hline  
 Uniform on $[a,b]$ & $\frac95$\\ 
 \hline
 $\mathsf{Beta}(\alpha,\beta)$ & $ 3 + \frac{6((\alpha-\beta)^2 (\alpha + \beta +1) - \alpha \beta (\alpha + \beta+2) )}{\alpha \beta (\alpha + \beta +2) (\alpha + \beta+3)}$\\
 \hline 
 $\Gamma(k,\theta)$ & $3 + \frac6k$\\ 
 \hline
 Maxwell-Boltzmann distribution & $3 +  4 \cdot \frac{40 \pi - 96 - 3\pi^2}{(3\pi - 8)^2}$\\
 \hline 
\end{tabular}
\caption{Some common $(c,2,4)$ hypercontractive random variables}
\end{table} 

The next proposition says that finitely supported distributions are $(c,2,4)$ hypercontractive where $c$ depends on the size of the smallest atom. 
\begin{proposition}~\label{prop:hyper-finite}
Let $\bX$ be supported over $\mathbb{R}$ and $\alpha = \min_{x: \bX(x) \not =0} \bX(x)$. Then, $\bX$ is $(c,4,2)$-hypercontractive where $c = \alpha^{-1/4}$. 
\end{proposition}
\begin{proof}
Without loss of generality, $\bX$ can be assumed to be centered i.e., $\mathbf{E}[\bX]=0$.  Let $y_1, \ldots, y_T$ be the support points of $\bX$ with mass $\beta_1, \ldots, \beta_T$. Thus, $\min_{i \in [T]} \beta_i =\alpha$. Thus, $\mu_4(\bX) = \sum_{i \in [T]} \beta_i y_i^4$ and $\mu_2(\bX) = \sum_{i \in [T]} \beta_i y_i^2$. 
Observe that, 
$$
\mathbf{E}[\bX^4] = \sum_{i \in [T]} \beta_i y_i^4 \le \bigg(\sum_{i \in [T]} \sqrt{\beta_i } y_i^2\bigg)^2 \le \frac{1}{\min_{i \in [T]} \beta_i} \cdot \bigg(\sum_{i \in [T]} \beta_i  y_i^2\bigg)^2 = \frac{1}{\alpha} \cdot \bigg(\sum_{i \in [T]} \beta_i  y_i^2\bigg)^2. 
$$
This concludes the proof. \end{proof}

\end{document}